\def\Title{Quantum Set Theory Extending the Standard Probabilistic Interpretation of Quantum Theory}

\def\Author{Masanao Ozawa}

\def\Instituteii{Nagoya University, Chikusa-ku, Nagoya,  464-8601, Japan}

\def\Email{ozawa@is.nagoya-u.ac.jp}

\def\Abstract{
The notion of equality between two observables will play many important roles in foundations of quantum theory.
However, the standard probabilistic interpretation based on the conventional 
Born formula does not give the probability of equality between two arbitrary observables, 
since the Born formula gives the probability distribution only for a commuting family of observables.
In this paper,  quantum set theory developed by Takeuti and the present author
is used to systematically extend the standard probabilistic interpretation 
of quantum theory to define the probability of equality between two arbitrary observables in an arbitrary state.
We apply this new interpretation to quantum measurement theory, 
and establish a logical basis for the difference between simultaneous measurability 
and simultaneous determinateness.}
\def\Keywords{quantum logic, quantum set theory, quantum theory, quantum measurements,
von Neumann algebras}

\mag\magstep 1
\documentclass[twoside,fleqn]{article}
\evensidemargin=\oddsidemargin
\usepackage{mathptmx}
\usepackage{latexsym}
\usepackage{enumerate}
\usepackage{amsmath,amssymb,amsthm,color,graphicx,bm}
\usepackage{bbm}

\newcommand{\af}{\,\et\,}
\newcommand{\al}{\alpha}
\newcommand{\be}{\beta}
\newcommand{\beq}{\begin{equation}}
\newcommand{\beqa}{\begin{eqnarray}}
\newcommand{\beqas}{\begin{eqnarray*}}
\newcommand{\beql}[1]{\begin{equation}\label{eq:#1}}
\newcommand{\bitem}{\begin{enumerate}[{\rm (i)}]\itemsep=0in \parskip=2pt}
\renewcommand{\L}{{\bf L}}
\newcommand{\bM}{{\bf M}}
\newcommand{\bP}{{\bf P}}
\newcommand{\bL}{{\bf L}}
\newcommand{\bra}[1]{\left\langle#1\right|}
\newcommand{\bracket}[1]{\left\langle#1\right\rangle}

\newcommand{\bS}{{\bf S}}
\newcommand{\bx}{{\bf x}}
\newcommand{\cA}{{\cal A}}
\newcommand{\cB}{{\cal B}}
\newcommand{\cC}{{\cal C}}
\newcommand{\cD}{\rm{dom}}

\newcommand{\cF}{{\cal F}}

\newcommand{\cH}{{\cal H}}
\newcommand{\ch}{\chi}
\newcommand{\cI}{{\cal I}}

\newcommand{\cK}{{\cal K}}

\newcommand{\cL}{{\cal L}}

\newcommand{\cM}{{\cal M}}
\newcommand{\cm}{{\rm com}_{o}}

\newcommand{\cO}{{\cal O}}
\newcommand{\com}{{\rm com}}
\newcommand{\commutes}{\,\rotatebox[origin=c]{270}{$\multimap$}\,}
\newcommand{\cut}[1]{}
\newcommand{\cP}{{\cal P}}
\newcommand{\cQ}{{\cal Q}}
\newcommand{\cR}{{\cal R}}
\newcommand{\cS}{{\cal S}}

\newcommand{\cuniv}{{\rm com}}

\newcommand{\cZ}{{\cal Z}}

\newcommand{\da}{\dagger}
\newcommand{\de}{\delta}
\newcommand{\De}{\Delta}
\newcommand{\dom}{{\rm dom}}
\newcommand{\eenum}{\end{enumerate}}
\newcommand{\eeq}{\end{equation}}
\newcommand{\eeqa}{\end{eqnarray}}
\newcommand{\eeqas}{\end{eqnarray*}}
\newcommand{\eitem}{\end{enumerate}}

\newcommand{\Eq}[1]{Eq.~(\ref{eq:#1})}
\newcommand{\eq}[1]{(\ref{eq:#1})}
\newcommand{\et}{\eta}

\newcommand{\Ga}{\Gamma}

\newcommand{\hu}{\hat{u}}

\newcommand{\id}{{\rm id}}
\newcommand{\Iff}{\leftrightarrow}
\newcommand{\IFF}{\Leftrightarrow}
\newcommand{\Inf}{\bigwedge}

\newcommand{\ignore}[1]{}

\newcommand{\ket}[1]{\left|#1\right\rangle}
\newcommand{\ketbra}[1]{\ket{#1}\!\!\bra{#1}}
\newcommand{\la}{\lambda}
\newcommand{\leo}{\le_{o}}
\newcommand{\leoo}{\le_{o'}}
\newcommand{\Lo}{\cL_{o}}
\newcommand{\Loo}{\cL_{o'}}

\newcommand{\M}{\bM}
\newcommand{\mb}{\mbox}
\newcommand{\nn}{\nonumber}
\newcommand{\Not}{\neg}
\newcommand{\om}{\omega}
\newcommand{\Om}{\Omega}

\newcommand{\Or}{\vee}
\newcommand{\p}{{}^{\perp}}

\newcommand{\ph}{\phi}
\newcommand{\Ph}{\Phi}

\newcommand{\ps}{\psi}
\newcommand{\Q}{{\bf Q}}
\newcommand{\R}{{\bf R}}
\newcommand{\ran}{\mbox{\rm ran}}
\newcommand{\rank}{\mbox{\rm rank}}

\newcommand{\rh}{\rho}
\newcommand{\RQ}{\R^{(\cQ)}}

\newcommand{\si}{\sigma}

\newcommand{\Sp}{{\rm Sp}}
\newcommand{\Sup}{\bigvee}

\newcommand{\tA}{\tilde{A}}
\newcommand{\tB}{\tilde{B}}

\newcommand{\Then}{\rightarrow}
\newcommand{\THEN}{\Rightarrow}
\newcommand{\Tr}{\mbox{\rm Tr}}

\newcommand{\tX}{\tilde{X}}
\newcommand{\tY}{\tilde{Y}}

\newcommand{\V}{V}
\newcommand{\val}[1]{[\![#1]\!]}
\newcommand{\valo}[1]{[\![#1]\!]_{o}}
\newcommand{\valoo}[1]{[\![#1]\!]_{o'}}

\newcommand{\VL}{\V^{(\cQ)}}

\newcommand{\VQ}{\V^{(\cQ)}}

\newcommand{\vval}[1]{[\![#1]\!]_{\cQ}}

\renewcommand{\And}{\wedge}
\renewcommand{\inf}{\bigwedge}

\renewcommand{\subset}{\subseteq}
\renewcommand{\sup}{\bigvee}

\newtheorem{theorem}{Theorem}[section]
\newtheorem{proposition}[theorem]{Proposition}
\newtheorem{lemma}[theorem]{Lemma}
\newcommand{\Cite}[1]{Ref.~\cite{#1}}

\newcommand{\deqs}[1]{ \begin{align*}#1\end{align*}}

  \title{ {\bf \Title}\thanks{An extended abstract of this paper was presented in
the 11th International Workshop on Quantum Physics and Logic (QPL 2014), 
Kyoto University, June 4--6, 2014 and appeared as Ref.~\cite{14QST}.}}
  \author{\sc \Author \\ \it\small \Instituteii\\ \tt\small \Email}
  \date{}

\begin{document}

\maketitle
  \thispagestyle{myheadings}
    \hspace{10\baselineskip}
\begin{abstract}
\Abstract 

\end{abstract}
\section{Introduction}
\label{se:1}
Set theory provides foundations of mathematics;
all the mathematical notions like numbers, functions, relations, and structures
are defined in the axiomatic set theory, 
ZFC  (Zermelo-Fraenkel set theory with the axiom of choice), and 
all the mathematical theorems are required to be provable in ZFC\citation{TZ71}.
Quantum set theory, instituted by Takeuti \cite{Ta81} and developed by the
present author \cite{07TPQ}, naturally extends the logical basis of set theory 
from classical logic to quantum logic to explore mathematics based on quantum
logic.

Despite remarkable success in axiomatic foundations of quantum mechanics \cite{Var85,Hol95},
the quantum logic approach to quantum foundations has not been 
considered powerful enough to solve interpretational problems \cite{Red87,Gib87}.
However, this weakness is considered to be mainly due to the fact that the 
conventional study of quantum logic has been limited to propositional logic.
Since quantum set theory extends the underlying logic from propositional logic to 
predicate logic, and provides set theoretical constructions of mathematical objects 
such as numbers, functions, relations, and structures based on quantum logic,  
we can expect that quantum set theory will provide much more systematic 
interpretation of quantum theory than the conventional quantum logic approach.
This paper represents the first step towards establishing systematic interpretation of
quantum theory based on quantum set theory, and naturally focusses on the most
fundamental notion in mathematics, namely, equality.

The notion of equality between quantum observables will play many important roles 
in foundations of quantum theory, in particular, in the theory of measurement and 
disturbance \cite{04URN,06QPC}.
However, the standard probabilistic interpretation based on the conventional Born formula
does not give the probability of equality between two arbitrary observables, 
since the Born formula gives the probability distribution only for a commuting family 
of observables \cite{vN55}.
In this paper,  quantum set theory is used to systematically extend the probabilistic 
interpretation of quantum theory to define the probability of equality between
two arbitrary observables in an arbitrary state based on the fact that real numbers defined 
in quantum set theory exactly corresponds to quantum observables \cite{Ta81,07TPQ}.  
It is shown that every observational proposition on a quantum system corresponds
to a statement in quantum set theory with the same projection-valued truth value
and the same probability in any state.
In particular, equality between real numbers in quantum set theory naturally 
provides a state-dependent notion of equality between quantum mechanical observables.
It has been broadly accepted that we cannot speak of the values of quantum observables 
without assuming a hidden variable theory, which are severely constrained by 
Kochen-Specker type no-go theorems \cite{KS67,Red87}. 
However, quantum set theory enables us to do so without assuming
hidden variables but alternatively with the consistent use of quantum logic. 
We apply this new interpretation to quantum measurement theory, 
and establish a logical basis for the difference between simultaneous measurability 
and simultaneous determinateness.

Section 2 provides preliminaries on complete orthomodular lattices,
commutators of their subsets, quantum logic on Hilbert spaces, and
the universe $\VQ$ of quantum set theory over a logic $\cQ$
on a Hilbert space $\cH$.
We give a characterization of the commutator of a subset
of a complete orthomodular lattice, improving Takeuti's characterization,
and give a factorization of the double commutant of a subset of a complete
orthomodular lattice into the maximal Boolean factor and a complete orthomodular
lattice without non-trivial Boolean factor.
Section 3 introduces a one-to-one correspondence obtained in Refs.~\cite{Ta81,07TPQ}
between the reals $\RQ$ in $\VQ$ and self-adjoint operators affiliated 
with the von Neumann algebra $\cM=\cQ''$ generated by $\cQ$, determines 
commutators and equality in $\RQ$, and gives the embedding of intervals in $\R$ into $\VQ$.
Section 4 formulates the standard probabilistic interpretation of quantum theory and 
also shows that the set of observational propositions for a quantum system can be embedded 
in a set of statements in quantum set theory without changing projection-valued truth value
assignment.  
Section 5 extends the standard interpretation by introducing simultaneous determinateness,
i.e., state-dependent commutativity of observables.  
We give several characterizations of simultaneous determinateness 
for finite number of quantum 
observables affiliated with an arbitrary von Neumann algebra in a given state, extending
some previous results  \cite{06QPC} on simultaneous determinateness for two observables.
Section 6 extends the standard interpretation by introducing quantum equality, i.e., 
state-dependent equality for two arbitrary observables. 
We give several characterizations of quantum equality for two
observables affiliated with an arbitrary von Neumann algebra in a given state, extending
some previous results  \cite{06QPC} on simultaneous determinateness for two observables.
 Sections 7 and 8 provide applications
to quantum measurement theory.  We discuss a state-dependent formulation of measurement 
of observables and simultaneous measurability, 
and establish a logical basis for the difference between simultaneous 
measurability and simultaneous determinateness.  The conclusion is given in Section 9.

Whereas we will discuss the completely general case where $\cM$ is an arbitrary von Neumann algebra, 
some results for the case where $\dim(\cH)<\infty$ and $\cM=\cB(\cH)$ have been
previously reported in \Cite{11QRM}.
In this special case, we can avoid the use of quantum set theory to introduce  simultaneous determinateness 
and quantum equality into the language of observational propositions, since simultaneous determinateness
and quantum equality can be expressed, respectively, by observational propositions 
constructed by atomic formulas of the form $X=x$ with an observable $X$ and a real number $x$.
However, to prove a transfer theorem ensuring that all the classical tautologies have the truth value 1, 
mentioned without proof in \Cite{11QRM}, Theorem 3, it is necessary, even in this special case, 
to develop quantum set theory and to define the embedding of the language of observational propositions 
into the language of quantum set theory.  The required machinery will be, for the first time, fully constructed 
in this paper including the case with observables with continuous spectrum, though the full power of this 
machinery will be revealed when applied to mathematical theorems beyond tautologies
after we have enriched the language of observational propositions, in the future research, 
with more sophisticated relations and functions than equality.

\section{Quantum set theory}
\subsection{Quantum logic}

A {\em complete orthomodular lattice}  is a complete
lattice $\cQ$ with an {\em orthocomplementation},
a unary operation $\perp$ on $\cQ$ satisfying
\bitem
\item[\rm (C1)]  if $P \le Q$ then $Q^{\perp}\le P^{\perp}$,
\item[\rm (C2)] $P^{\perp\perp}=P$,
\item[\rm  (C3)] $P\Or P^{\perp}=1$ and $P\And P^{\perp}=0$,
where $0=\Inf\cQ$ and $1=\Sup\cQ$,
\eitem
that  satisfies the {\em orthomodular law}
\bitem
\item[\rm (OM)] if $P\le Q$ then $P\Or(P^{\perp}\And Q)=Q$.
\eitem
In this paper, any complete orthomodular lattice is called a {\em logic}.
A non-empty subset of a logic $\cQ$ is called a {\em subalgebra} iff
it is closed under $\And $, $\Or$, and $\perp$.
A subalgebra $\cA$ of $\cQ$ is said to be {\em complete} iff it has
the supremum and the infimum in $\cQ$ of an arbitrary subset of $\cA$.
For any subset $\cA$ of $\cQ$, 
the subalgebra generated by $\cA$ is denoted by
$\Ga_0\cA$.
We refer the reader to Kalmbach \cite{Kal83} for a standard text on
orthomodular lattices.

We say that $P$ and $Q$ in a logic $\cQ$
{\em commute}, in  symbols
$P\commutes Q$, iff  $P=(P\And Q)\Or(P\And
Q^{\perp})$. All the relations $P\commutes Q$,
$Q\commutes P$,
$P^{\perp}\commutes Q$, $P\commutes Q^{\perp}$,
and $P^{\perp}\commutes Q^{\perp}$ are equivalent.
The distributive law does not hold in general, but
the following useful propositions hold (\Cite{Kal83}, pp.~24--25).

\begin{proposition}\label{th:distributivity}
If $P_1,P_2\commutes
Q$, then the sublattice generated by $P_1,P_2,Q$ is
distributive.
\end{proposition}

\begin{proposition}\label{th:logic}
If $P_{\al}\commutes Q$ for all $\al$, then
$\Sup_{\al}P_{\al}\commutes Q$, 
$\Inf_{\al}P_{\al}\commutes Q$,
$Q \And (\Sup_{\al}P_{\al})=\Sup_{\al}(Q\And
P_{\al})$,
and 
$Q \Or (\Inf_{\al}P_{\al})=\Inf_{\al}(Q\Or
P_{\al})$,
\end{proposition}
From Proposition \ref{th:distributivity}, a logic $\cQ$ is a Boolean
algebra if and only if $P\commutes Q$  for all $P,Q\in\cQ$
(Ref.~\cite[pp.~24--25]{Kal83}).

For any subset $\cA\subseteq\cQ$,
we denote by $\cA^{!}$ the {\em commutant} 
of $\cA$ in $\cQ$ (\Cite{Kal83}, p.~23), i.e.,
\[
\cA^{!}=
\{P\in\cQ\mid P\commutes Q \mbox{ for all }
Q\in\cA\}.
\]
Then, $\cA^{!}$ is a complete subalgebra of $\cQ$.
A {\em sublogic} of $\cQ$ is a subset $\cA$ of
$\cQ$ satisfying $\cA=\cA^{!!}$. 
For any subset $\cA\subseteq\cQ$, the smallest 
logic including $\cA$ is 
$\cA^{!!}$ called the  {\em sublogic generated by
$\cA$}.
Then, it is easy to see that a subset 
 $\cA$ is a Boolean sublogic, or equivalently 
 a distributive sublogic, if and only if 
$\cA=\cA^{!!}\subseteq\cA^{!}$.

\subsection{Commutators}
\label{se:CIQL}

Let $\cQ$ be a logic.
Marsden \cite{Mar70} has introduced the commutator $\com(P,Q)$ 
of two elements $P$ and $Q$ of $\cQ$ by 
\beqa
\com(P,Q)=(P\And Q)\Or(P\And Q\p)\Or(P\p\And Q)\Or(P\p\And Q\p).
\eeqa
Bruns and Kalmbach \cite{BK73} have generalized this notion 
to finite subsets of $\cQ$ by 
\beq
\com(\cF)=\Sup_{\al:\cF\to\{\id,\perp\}}\Inf_{P\in\cF}P^{\al(P)}
\eeq
for all $\cF\in\cP_{\om}(\cQ)$,
where $\cP_{\om}(\cQ)$ stands for the set of finite subsets of $\cQ$, and
$\{\id,\perp\}$ stands for the set consisting of the identity operation $\id$ and the 
orthocomplementation~$\perp$.
Generalizing this notion to arbitrary subsets $\cA$ of $\cQ$, Takeuti \cite{Ta81} defined
$\com(\cA)$ by
\beqa
\com(\cA)&=&\Sup T(\cA),\\
T(\cA)&=&\{E\in\cA^{!} \mid P_{1}\And E\commutes P_{2}\And E
\mb{ for all }P_{1},P_{2}\in\cA\},
\eeqa
of any $\cA\in\cP(\cQ)$, where $\cP(\cQ)$ stands for the power set of $\cQ$,
and showed that $\com(\cA)\in T(\cA)$.
Subsequently, Pulmannov\'{a} \cite{Pul85} showed:
\begin{theorem}
For any subset $\cA$ of a logic $\cQ$, we have
\bitem
\item[\rm (i)] $\com(\cA)=\Inf\{\com(\cF)\mid \cF\in\cP_{\om}(\cA)\},$
\item[\rm (ii)] $\com(\cA)=\Inf\{\com(P,Q)\mid P,Q\in \Ga_0(\cA)\}$.
\eitem
\end{theorem}

Here, we reformulate Takeuti's definition in a more convenient form.
Let $\cA\subseteq\cQ$.
Note that $\cA^{!!}$ is the sublogic generated by $\cA$,
and $\cA^{!}\cap\cA^{!!}$ is the center of $\cA^{!!}$, i.e., 
the set of elements of $\cA^{!!}$
commuting with all elements of $\cA^{!!}$.
Denote by $L(\cA)$ the sublogic generated by $\cA$, i.e., 
$L(\cA)=\cA^{!!}$,
and by $Z(\cA)$ the center of $L(\cA)$, i.e., $Z(\cA)=\cA^{!}\cap\cA^{!!}$.
A {\em subcommutator} of $\cA$ is  any $E\in Z(\cA)$ 
such that $P_1\And E\commutes P_2\And E$ for all $P_1,P_2\in\cA$.
Denote by $S(\cA)$ the set of subcommutators of  $\cA$, i.e., 
\beqa
S(\cA)=\{E\in Z(\cA)\mid P_{1}\And E\commutes P_{2}\And E
\mb{ for all }P_{1},P_{2}\in\cA\}.
\eeqa
By the relation $Z(\cA)\subseteq\cA^{!}$, we immediately obtain
the relation
$
\Sup S(\cA)\le\com(\cA).
$
We shall show that the equality actually holds.

\begin{lemma}\label{th:subcommutator}
Let $\cA$ be any subset of a logic $\cQ$.
For any  $P_1,P_2\in\cA$ and $E\in \cA^{!}$,
we have $P_1\And E\commutes P_2\And E$
if  and only if $P_1\And E \commutes P_2$.
\end{lemma}
\begin{proof}
Let $E\in \cA^{!}$ and $P_1,P_2\in\cA$.
We have
$
(P_{1}\And E)\And (P_{2}\And E)^{\perp}=
(P_{1}\And E)\And P_{2}^{\perp},
$
and hence
\[
[(P_1\And E)\And (P_{2}\And E)]\Or[(P_1\And E)\And (P_{2}\And
E)^{\perp}]=[(P_1\And E)\And P_2]\Or[(P_1\And E)\And P_2^{\perp}].
\]
It follows that
$P_1\And E\commutes P_2\And E$ if and only if
$P_1\And E\commutes P_2$.
\end{proof}

For any $P,Q\in\cQ$, the {\em interval} $[P,Q]$ is the set of
all $X\in\cQ$ such that $P\le X\le Q$.
For any $\cA\subseteq \cQ$ and $P,Q\in\cA$,
we write $[P,Q]_{\cA}=[P,Q]\cap\cA$.

\begin{theorem}\label{th:equivalence_com}
For any subset  $\cA$  of a logic $\cQ$, the following relations hold.
\bitem
\item[\rm (i)] $S(\cA)=\{E\in Z(\cA)\mid [0,E]_{\cA}\subseteq Z(\cA)\}$.
\item[\rm (ii)] $\Sup S(\cA)$ is the maximum subcommutator of $\cA$, i.e., $\Sup S(\cA)\in S(\cA)$.
\item[\rm (iii)] $S(\cA)=[0,\Sup S(\cA)]_{L(\cA)}$.
\item[\rm (iv)] $\com(\cA)=\Sup S(\cA)$.
\eitem
\end{theorem}
\begin{proof}
(i) It is easy to see that $P_1\And E\commutes P_2$ for every $P_1,P_2\in\cA$
if and only if $[0,E]\cap\cA\subseteq\cA^{!}$, and hence the assertion 
follows from Lemma \ref{th:subcommutator}.
(ii) Let $P_1,P_2\in\cA$.
We have $P_1\And E\commutes P_2$ for every
$E \in S(\cA)$ from Lemma \ref{th:subcommutator},
and $P_1\And \Sup S(\cA)\commutes P_2$
from Proposition \ref{th:logic}.  Since $S(\cA)\subseteq Z(\cA)$,
we have $\Sup S(\cA)\in Z(\cA)$.  Thus, $\Sup S(\cA)\in S(\cA)$,
and the assertion follows.
(iii)  If $P\in [0,\Sup S(\cA)]_{L(\cA)}$ then $P=P\And \Sup S(\cA)$ commutes with
every element of $L(\cA)$.  Thus, we have $[0,\Sup S(\cA)]_{L(\cA)}=
[0,\Sup S(\cA)]_{Z(\cA)}$.
Now, let $P\in [0,\Sup S(\cA)]_{Z(\cA)}$. 
Then, $P_1\commutes P$  and $P_1\commutes P_2\And \Sup S(\cA)$,
and hence $P_1\commutes P\And P_2\And \Sup S(\cA)$
and
$P_1\commutes P_2\And P$.
Thus, we have $P\in S(\cA)$,
and the assertion follows.
(iv) Since $\com(\cF)\in Z(\cF)$ for every finite subset $\cF$  of $\cA$, 
we have $\com(\cA)\in Z(\cA)$,
and hence we have $\com(\cA)\in Z(\cA)$.  
Thus, relation (iv) follows.
\end{proof}
\color{black}

The following proposition will be useful in later discussions.

\begin{theorem}\label{th:maximal_Boolean}
Let $\cB$ be a maximal Boolean sublogic of a logic $\cQ$ and $\cA$ a
subset of $\cQ$ including $\cB$, i.e., $\cB\subseteq\cA\subseteq\cQ$.
Then, we have $\com(\cA)\in\cB$ and $[0,\com(\cA)]_{\cA}\subset \cB$.
\end{theorem}
\begin{proof}
Since $\com(\cA)\in Z(\cA)\subseteq\cB^{!}=\cB$, we have $\com(\cA)\in\cB$.  Let $P\in\cA$.  
Then, $P\And \com(\cA)\commutes Q$
for all $Q\in\cB$, so that $P\And \com(\cA)\in\cB^{!}=\cB$, and hence 
$[0,\com(\cA)]_{\cA}\subset \cB$.
\end{proof}

The following theorem clarifies the significance of commutators.

\begin{theorem}\label{th:absoluteness_commutator}
Let $\cA$ be a subset of a logic $\cQ$.
Then,  $L(\cA)$ is isomorphic to the direct product of the complete Boolean algebra
$[0,\com(\cA)]_{L(\cA)}$ and the complete orthomodular lattice
$[0,\com(\cA)^\perp]_{L(\cA)}$ without non-trivial Boolean factor.
\end{theorem}
\sloppy
\begin{proof}
It follows from $\Sup S(\cA)\in Z(\cA)$ that
$L(\cA)\cong [0,\Sup S(\cA)]_{L(\cA)}\times [0,\Sup S(\cA)^\perp]_{L(\cA)}$.
Then, $[0,\Sup S(\cA)]_{L(\cA)}$ is a complete Boolean algebra,
since $[0,\Sup S(\cA)]_{L(\cA)}\subseteq Z(\cA)$.
It follows easily from the maximality of $\Sup S(\cA)$ that 
$[0,\Sup S(\cA)^\perp]_{L(\cA)}$ has no non-trivial Boolean factor.
Thus, the assertion follows from the relation $\Sup S(\cA)=\com(\cA)$.
\end{proof}

We refer the reader to Pulmannov\'{a} \cite{Pul85} and 
Chevalier \cite{Che89} for further results about commutators
in orthomodular lattices. 

\subsection{Logic on Hilbert spaces}
\label{se:QL}
Let $\cH$ be a Hilbert space.
For any subset $S\subseteq\cH$,
we denote by $S^{\perp}$ the orthogonal complement
of $S$.
Then, $S^{\perp\perp}$ is the closed linear span of $S$.
Let $\cC(\cH)$ be the set of all closed linear subspaces in
$\cH$. 
With the set inclusion ordering, 
the set $\cC(\cH)$ is a complete
lattice. 
The operation $M\mapsto M^\perp$ 
is  an orthocomplementation
on the lattice $\cC(\cH)$, with which $\cC(\cH)$ is a logic.

Denote by $\cB(\cH)$ the algebra of bounded linear
operators on $\cH$ and $\cQ(\cH)$ the set of projections on $\cH$.
We define the {\em operator ordering} on $\cB(\cH)$ by
$A\le B$ iff $(\ps,A\ps)\le (\ps,B\ps)$ for
all $\ps\in\cH$. 
For any $A\in\cB(\cH)$, denote by $\cR (A)\in\cC(\cH)$
the closure of the range of $A$, {\em i.e.,} 
$\cR(A)=(A\cH)^{\perp\perp}$.
For any $M\in\cC(\cH)$,
denote by $\cP (M)\in\cQ(\cH)$ the projection operator 
of $\cH$
onto $M$.
Then, $\cR\cP (M)=M$ for all $M\in\cC(\cH)$
and $\cP\cR (P)=P$ for all $P\in\cQ(\cH)$,
and we have $P\le Q$ if and only if $\cR (P)\subseteq\cR (Q)$
for all $P,Q\in\cQ(\cH)$,
so that $\cQ(\cH)$ with the operator ordering is also a logic
isomorphic to $\cC(\cH)$.
Any sublogic of $\cQ(\cH)$ will be called a {\em logic on $\cH$}. 
The lattice operations are characterized by 
$P\And Q={\mb{weak-lim}}_{n\to\infty}(PQ)^{n}$, 
$P^\perp=1-P$ for all $P,Q\in\cQ(\cH)$.

Let $\cA\subseteq\cB(\cH)$.
We denote by $\cA'$ the {\em commutant of 
$\cA$ in $\cB(\cH)$}.
A self-adjoint subalgebra $\cM$ of $\cB(\cH)$ is called a
{\em von Neumann algebra} on $\cH$ iff 
$\cM''=\cM$.
For any self-adjoint subset $\cA\subseteq\cB(\cH)$,
$\cA''$ is the von Neumann algebra generated by $\cA$.
We denote by $\cP(\cM)$ the set of projections in
a von Neumann algebra $\cM$.
For any $P,Q\in\cQ(\cH)$, we have 
$P\commutes Q$ iff $[P,Q]=0$, where $[P,Q]=PQ-QP$.
For any subset $\cA\subseteq\cQ(\cH)$,
we denote by $\cA^{!}$ the {\em commutant} 
of $\cA$ in $\cQ(\cH)$.
For any subset $\cA\subseteq\cQ(\cH)$, the smallest 
logic including $\cA$ is the logic
$\cA^{!!}$ called the  {\em logic generated by
$\cA$}.  
Then, a subset $\cQ \subseteq\cQ(\cH)$ is a logic on $\cH$ if
and only if $\cQ=\cP(\cM)$ for some von Neumann algebra
$\cM$ on $\cH$ (\Cite{07TPQ}, Proposition 2.1).

We define the {\em implication} and the 
{\em logical equivalence} on $\cQ$ by  
$P\Then Q=P^{\perp}\Or(P\And Q)$
 and
$P\Iff Q=(P\Then Q)\And(Q\Then P)$.
We have the following characterization of commutators in 
logics on Hilbert spaces (\Cite{07TPQ}, Theorems 2.5, 2.6).

\begin{theorem}\label{th:com}
Let $\cQ$ be a logic on $\cH$
and let $\cA\subseteq\cQ$.
Then, we have the following relations.
\bitem
\item
$\com(\cA)=\cP\{\ps\in\cH\mid [A,B]\ps=0 \mb{ for all }A,B\in\cA''\}$.
\item
$\com(\cA)=\cP\{\ps\in\cH\mid 
[P_1,P_2]P_3\psi=0\mb{ for all } P_1,P_2,P_3\in \cA\}$.
\eitem
\end{theorem}

\subsection{Quantum set theory over logic on Hilbert spaces}
\label{se:UQ}

We denote by $\V$ the universe 
of the Zermelo-Fraenkel set theory
with the axiom of choice (ZFC).
Let $\cL(\in)$ be the first-order language with equality without constant symbols
augmented by a binary relation symbol
$\in$, bounded quantifier symbols $\forall x\in y$,
$\exists x \in y$
(in addition to unbounded quantifier symbols $\forall x$,
$\exists x$.
For any class $U$, 
the language $\cL(\in,U)$ is the one
obtained by adding a name for each element of $U$.

Let $\cQ$ be a logic on $\cH$.
For each ordinal $ {\al}$, let
\beq
\V_{\al}^{(\cQ)} = \{u|\ u:\dom(u)\to \cQ \mbox{ and }
(\exists \be<\al)
\dom(u) \subseteq V_{\be}^{(\cQ)}\}.
\eeq
The {\em $\cQ$-valued universe} $\VL$ is defined
by 
\beq
  \VL= \bigcup _{{\al}{\in}\mbox{On}} V_{{\al}}^{(\cQ)},
\eeq
where $\mbox{On}$ is the class of all ordinals. 
For every $u\in\VQ$, the rank of $u$, denoted by
$\rank(u)$,  is defined as the least $\al$ such that
$u\in \VQ_{\al+1}$.
It is easy to see that if $u\in\dom(v)$ then 
$\rank(u)<\rank(v)$.

For any $u,v\in\VL$, the $\cQ$-valued truth values of
atomic formulas $u=v$ and $u\in v$ are assigned
by the following rules recursive in rank.
\bitem
\item $\vval{u = v}
= \inf_{u' \in  \cD(u)}(u(u') \Then
\vval{u'  \in v})
\And \inf_{v' \in   \cD(v)}(v(v') 
\Then \vval{v'  \in u})$.
\item $ \vval{u \in v} 
= \sup_{v' \in \cD(v)} (v(v')\And \vval{u =v'})$.
\eitem

To each statement $\ph$ of $\cL(\in,\VL)$ we assign the
$\cQ$-valued truth value $ \val{\ph}_{\cQ}$ by the following
rules.
\bitem
\setcounter{enumi}{2}
\item $ \vval{\Not\ph} = \vval{\ph}^{\perp}$.
\item $ \vval{\ph_1\And\ph_2} 
= \vval{\ph_{1}} \And \vval{\ph_{2}}$.
\item $ \vval{\ph_1\Or\ph_2} 
= \vval{\ph_{1}} \Or \vval{\ph_{2}}$.
\item $ \vval{\ph_1\Then\ph_2} 
= \vval{\ph_{1}} \Then \vval{\ph_{2}}$.
\item $ \vval{\ph_1\Iff\ph_2} 
= \vval{\ph_{1}} \Iff \vval{\ph_{2}}$.
\item $ \vval{(\forall x\in u)\, {\ph}(x)} 
= \Inf_{u'\in \dom(u)}
(u(u') \Then \vval{\ph(u')})$.
\item $ \vval{(\exists x\in u)\, {\ph}(x)} 
= \Sup_{u'\in \dom(u)}
(u(u') \And \vval{\ph(u')})$.
\item $ \vval{(\forall x)\, {\ph}(x)} 
= \Inf_{u\in \VL}\vval{\ph(u)}$.
\item $ \vval{(\exists x)\, {\ph}(x)} 
= \Sup_{u\in \VL}\vval{\ph(u)}$.
\eitem

We say that a statement ${\ph}$ of $ \cL(\in,\VL) $
{\em holds} in $\VL$ iff $ \val{{\ph}}_{\cQ} = 1$.
A formula in $\cL(\in)$ is called a {\em
$\De_{0}$-formula}  iff it has no unbounded quantifiers
$\forall x$ or $\exists x$.
The following theorem holds \cite{07TPQ}.

\sloppy
\begin{theorem}[$\De_{0}$-Absoluteness Principle]
\label{th:Absoluteness}
\sloppy  
For any $\De_{0}$-formula 
${\ph} (x_{1},{\ldots}, x_{n}) $ 
of $\cL(\in)$ and $u_{1},{\ldots}, u_{n}\in \VQ$, 
we have
\[
\val{\ph(u_{1},\ldots,u_{n})}_{\cQ}=
\val{\ph(u_{1},\ldots,u_{n})}_{\cQ(\cH)}.
\]
\end{theorem}
Henceforth, 
for any $\De_{0}$-formula 
${\ph} (x_{1},{\ldots}, x_{n}) $
and $u_1,\ldots,u_n\in\VQ$,
we abbreviate $\val{\ph(u_{1},\ldots,u_{n})}=
\val{\ph(u_{1},\ldots,u_{n})}_{\cQ}$,
which is the common $\cQ$-valued truth value 
in all $\VL$ such  that $u_{1},\ldots,u_{n}\in\VL$.

The universe $\V$  can be embedded in
$\VQ$ by the following operation 
$\vee:v\mapsto\check{v}$ 
defined by the $\in$-recursion: 
for each $v\in\V$, $\check{v} = \{\check{u}|\ u\in v\} 
\times \{1\}$. 
Then we have the following \cite{07TPQ}.
\begin{theorem}[$\De_0$-Elementary Equivalence Principle]
\label{th:2.3.2}
\sloppy  
For any $\De_{0}$-formula 
${\ph} (x_{1},{\ldots}, x_{n}) $ 
of $\cL(\in)$ and $u_{1},{\ldots}, u_{n}\in V$,
we have
$
\bracket{\V,\in}\models  {\ph}(u_{1},{\ldots},u_{n})
\mbox{ if and only if }
\val{\ph(\check{u}_{1},\ldots,\check{u}_{n})}=1.
$
\end{theorem}

For $u\in\VQ$, we define the {\em support} 
of $u$, denoted by $\L(u)$, by transfinite recursion on the 
rank of $u$ by the relation
\beq
\L(u)=\bigcup_{x\in\dom(u)}\L(x)\cup\{u(x)\mid x\in\dom(u)\}.
\eeq
For $\cA\subseteq\VQ$ we write 
$\L(\cA)=\bigcup_{u\in\cA}\L(u)$ and
for $u_1,\ldots,u_n\in\VQ$ we write 
$\L(u_1,\ldots,u_n)=\L(\{u_1,\ldots,u_n\})$.
Let $\cA\subseteq\VQ$.  The {\em commutator
of $\cA$}, denoted by $\com(\cA)$, is defined by 
\beq
\cuniv(\cA)=\com (\L(\cA)).
\eeq
For any $u_1,\ldots,u_n\in\VQ$, we write
$\cuniv(u_1,\ldots,u_n)=\cuniv(\{u_1,\ldots,u_n\})$.
For bounded theorems,  the following transfer 
principle holds \cite{07TPQ}.

\begin{theorem}[ZFC  Transfer Principle]
\label{th:TP}
For any $\De_{0}$-formula ${\ph} (x_{1},{\ldots}, x_{n})$ 
of $\cL(\in)$ and $u_{1},{\ldots}, u_{n}\in\VQ$, if 
${\ph} (x_{1},{\ldots}, x_{n})$ is provable in ZFC, then
we have
\[
\cuniv(u_{1},\ldots,u_{n})\le
\val{\ph({u}_{1},\ldots,{u}_{n})}.
\]
\end{theorem}

\section{Real numbers in quantum set theory}
\label{se:RN}

Let $\Q$ be the set of rational numbers in $V$.
We define the set of rational numbers in the model $\VQ$
to be $\check{\Q}$.
We define a real number in the model by a Dedekind cut
of the rational numbers. More precisely, we identify
a real number with the upper segment of a Dedekind cut
assuming that the lower segment has no end point.
Therefore, the formal definition of  the predicate $\R(x)$, 
``$x$ is a real number,'' is expressed by
\beqa
\R(x)&:=&
\forall y\in x(y\in\check{\Q})
\And \exists y\in\check{\Q}(y\in
x)
\And \exists y\in\check{\Q}(y\not\in x)\nn\\
& &  \And
\forall y\in\check{\Q}(y\in x\Iff\forall z\in\check{\Q}
(y<z \Then z\in x)).
\eeqa
The symbol ``:='' is used to define a new formula, here and hereafter.
We define $\R^{(\cQ)}$ to be the interpretation of 
the set $\R$ of real
numbers in $\VQ$ as follows.
\beq
\R^{(\cQ)} = \{u\in\VQ|\ \cD(u)=\cD(\check{\Q})
\ \mb{and }\val{\R(u)}=1\}.
\eeq
The set $\R_\cQ$ of real numbers in $\VQ$ is defined by
\beq
\R_\cQ=\R^{(\cQ)}\times\{1\}.
\eeq
Then, for any $u,v\in\R^{(\cQ)}$,
the following relations hold in $\VQ$ \cite{07TPQ}.
\bitem
\item $\val{(\forall u\in\R_\cQ) u=u}=1$.
\item $\val{(\forall u,v\in\R_\cQ) u=v\Then v=u}=1$.
\item $\val{(\forall u,v,w\in\R_\cQ) u=v\And v=w\Then u=w}=1$.
\item $\val{(\forall v\in\R_\cQ)(\forall x,y\in v)x=y\And x\in v\Then
y\in v}$.
\item $\val{(\forall u,v\in\R_\cQ)(\forall x\in u)x\in u\And u=v\Then x\in
v}.$
\eitem
From the above, 
the equality is an equivalence relation between real numbers in $\VQ$.
For any $u_1,\ldots,u_n\in\R^{(\cQ)}$, we have
\beq
\val{u_1=u_2\And\cdots\And u_{n-1}=u_n}\le 
\cuniv(u_1,\ldots,u_n),
\eeq
and hence commutativity follows from equality in $\R^{(\cQ)}$ \cite{07TPQ}.

Let $\cM$ be a von Neumann algebra on a Hilbert
space $\cH$ and let $\cQ=\cP(\cM)$.
A closed operator $A$ (densely defined) on $\cH$ is
said to be {\em affiliated} with $\cM$, in symbols $A\,\et\,\cM$, 
iff $U^{*}AU=A$ for any unitary operator $U\in\cM'$.
Let $A$ be a self-adjoint operator (densely defined) on $\cH$
and let $A=\int_{\R} \la\, dE^A(\la)$
be its spectral decomposition, where $\{E^{A}(\la)\}_{\la\in\R}$ 
is the resolution of  identity belonging to $A$ (\Cite{vN55}, p.\ 119).
It is well-known that $A\af\cM$ if and only if $E^A({\la})\in\cQ$ 
for every $\la\in\R$.
Denote by $\overline{\cM}_{SA}$ the set of self-adjoint operators 
affiliated with $\cM$.
Two self-adjoint operators $A$ and $B$ are said to {\em commute},
in symbols $A\commutes B$,
iff $E^A(\la)\commutes E^B(\la')$ for every pair 
$\la,\la'$ of reals.

For any $u\in\R^{(\cQ)}$ and $\la\in\R$, we define $E^{u}(\la)$ by 
\beq
E^{u}(\la)=\Inf_{\la<r\in\Q}u(\check{r}).
\eeq
Then, it can be shown that
$\{E^u(\la)\}_{\la\in\R}$ is a resolution of
identity in $\cQ$ and hence by the spectral theorem there
is a self-adjoint operator $\hat{u}\af\cM$ uniquely
satisfying $\hat{u}=\int_{\R}\la\, dE^u(\la)$.  On the other
hand, let $A\af\cM$ be a self-adjoint operator. We define $\tilde{A}\in\VQ$ by
\beq
\tilde{A}=\{(\check{r},E^{A}(r))\mid r\in\Q\}.
\eeq
Then, $\dom(\tA)=\dom(\Q)$ and $\tA(\check{r})=E^{A}(r)$ for all $r\in\Q$.
It is easy to see that $\tilde{A}\in\RQ$ and we have
$(\hat{u})\tilde{}=u$ for all $u\in\RQ$ and $(\tilde{A})\hat{}=A$
for all $A\in\overline{\cM}_{SA}$.
\cut{
\[
\dom(\tilde{A})=\dom(\check{\Q})\mbox{ and }
\tilde{A}(\check{r})=E^{A}(r)\mbox{ for all }r\in\Q.
\]
Then, it is easy to see that $\tilde{A}\in\RQ$ and we have
$(\hat{u})\tilde{}=u$ for all $u\in\RQ$ and $(\tilde{A})\hat{}=A$
for all $A\in\overline{\cM}_{SA}$.
}
Therefore, the correspondence
between $\RQ$ and $\overline{\cM}_{SA}$ is a one-to-one correspondence.
We call the above correspondence the {\em Takeuti correspondence}.
Now, we have the following \cite{07TPQ}.

\begin{theorem}
Let $\cQ$ be a logic on $\cH$.  The relations 
\bitem
\item ${\displaystyle E^{A}(\la)=\Inf_{\la<r\in\Q}u(\check{r})}$ for all $\la\in\Q$,
\item $u(\check{r})=E^{A}(r)$ for all $r\in\Q$,
\eitem
for all $u=\tA\in\RQ$ and $A=\hu\in \overline{\cM}_{SA}$ 
sets up a one-to-one correspondence between $\RQ$ and $ \overline{\cM}_{SA}$.
\end{theorem}

For any $r\in\R$, we shall write $\tilde{r}=(r1)\,\tilde{}$, where $r1$ is the scalar 
operator on $\cH$.
Then, we have $\dom(\tilde{r})=\dom(\check{\Q})$ and $\tilde{r}(\check{t})=
\val{\check{r}\le \check{t}}$, so that we have $\L(\tilde{r})=\{0,1\}$.
Denote by $\cB(\R^n)$ the $\si$-filed of Borel subsets of $\R^n$ and $B(\R^n)$ the space
of bounded Borel functions on $\R^n$.
A {\em spectral measure} \cite{Hal51} on $\R^{n}$ in $\cM$ is a mapping $E$ 
of $\cB(\R^{n})$ into $\cP(\cM)$ satisfying $\sum_{j}E(\De_i)=1$ for any disjoint sequence  
$\{\De_j\}$ in $\cB(\R^{n})$ such that $\bigcup_{j}\De_j=\R^{n}$.  
Let $X$ be a self-adjoint operator affiliated with $\cM$.
For any $f\in B(\R)$, the bounded self-adjoint operator $f(X)\in\cM$ is
defined by $f(X)=\int_{\R}f(\la) dE^{X}(\la)$.
The {\em spectral measure of $X$} is a spectral measure $E^{A}$ on $\R$ in $\cM$ 
defined by $E^{X}(\De)=\ch_{\De}(X)$ for any $\De\in\cB(\R)$.
Then, we have $E^{X}(\la)=E^{X}((-\infty,\la])$.  

\begin{proposition}\label{th:QBorel}
Let $r\in\R$, $s,t\in\R$, and $X\af\cM_{SA}$.
We have the following relations.
\bitem
\item $\val{\check{r}\in\tilde{s}}=\val{\check{s}\le \check{r}}
=E^{s1}(t)$.
\item $\val{\tilde{s}\le\tilde{t}}=\val{\check{s}\le\check{t}}
=E^{s1}(t)$.
\item $\val{\tilde{X}\le\tilde{t}}=E^{X}(t)=E^{X}((-\infty,t])$.
\item $\val{\tilde{t}<\tilde{X}}=1-E^{X}(t)=E^{X}((t,\infty))$.
\item$\val{\tilde{s}<\tilde{X}\le \tilde{t}}=E^X({t})-E^X({s})=
E^{X}((s,t])$.
\item $\val{\tilde{X}=\tilde{t}}
=E^{X}(t)-\Sup_{r<t,r\in\Q} E^{X}(r)=E^{X}(\{t\})$.
\eitem
\end{proposition}
\begin{proof}
Relations (i), (ii), and (iii) follows from \cite[Proposition 5.11]{07TPQ}.
We have $\com(\tilde{t},\tilde{X})=1$, so that (iv) follows from the 
ZFC Transfer Principle (Theorem \ref{th:TP}).
Relation (v) follows from (iii) and (iv).
We have
\beqas
\val{\tilde{X}=\tilde{t}}&=&
\Inf_{r\in\Q}\tilde{X}(\check{r})\Then
\val{\check{r}\in\tilde{t}}\And
\Inf_{r\in\Q}\tilde{t}(\check{r})\Then
\val{\check{r}\in\tilde{X}}\\
&=&
\Inf_{r\in\Q}
E^{X}(r)^{\perp}\Or E^{t1}(r)
\And
\Inf_{r\in\Q}E^{t1}(r)^\perp
\Or E^{X}(r)\\
&=&
\Inf_{r<t, r\in\Q}
E^{X}(r)^{\perp}
\And
\Inf_{t\le r\in\Q} E^{X}(r)\\
&=&
[1-\Sup_{r<t, r\in Q}E^{X}(r)]
\And E^{X}(t)\\
&=&
E^{X}(t)- \Sup_{r<t,r\in\Q} E^{X}(r)\\
&=&
E^{X}(\{t\}).
\eeqas
Thus, relation (vi) follows.
\end{proof}

\section{Standard probabilistic interpretation of quantum theory} 
\label{se:2}

Let $\bS$ be a quantum system described by a von Neumann algebra $\cM$
on a Hilbert space $\cH$.
According to the standard formulation of quantum theory,
the {\em observables} of $\bS$ are defined as self-adjoint operators 
affiliated with $\cM$,
the {\em states} of $\bS$ are represented by density operators on $\cH$, 
and a {\em vector state} $\psi$ is identified with the state $\ketbra{\psi}$.
We denote by $\cO(\cM)$ the set of observables,
by $\cS(\cH)$ the space of density operators.
Observables $X_1,\ldots,X_n\in\cO(\cM)$ are said to be {\em mutually commuting} iff
$X_j\commutes X_k$ for all $j,k=1,\ldots,n$.
If $X_1,\ldots,X_n\in\cO(\cM)$ are bounded, this condition is equivalent to 
$[X_j,X_k]=0$ for all $j,k=1,\ldots,n$.
The standard probabilistic interpretation of quantum theory defines the 
{\em joint probability distribution function} $F^{X_1,\ldots,X_n}_{\rh}(x_1,\ldots,x_n)$
for mutually commuting observables $X_1,\ldots,X_n\in\cO(\cM)$ in 
$\rh\in\cS(\cH)$ by the {\em Born statistical formula}:
\beq
F^{X_1,\ldots,X_n}_{\rh}(x_1,\ldots,x_n)=\Tr[E^{X_1}(x_1)\cdots E^{X_n}(x_n)\rh].
\eeq

To clarify the logical structure presupposed in the standard probabilistic interpretation,
we define {\em observational propositions} for $\bS$ 
by the following rules.
\bitem
\item[(R1)] For any $X\in\cO(\cM)$ and $x\in \R$, the expression
$X\leo x$ is an observational proposition.
\item[(R2)] If $\ph_1$ and $\ph_2$ are observational propositions,
$\Not \ph_1$ and $\ph_1\And \ph_2$ are also observational propositions.
\eitem
Thus, every observational proposition is built up from ``atomic'' 
observational propositions $X\leo x$ by adding finite number of
connectives $\Not$ and $\And$.
We denote by $\Lo(\cM)$ the set of observational propositions.
We introduce the connective $\Or$ by definition.
\bitem
\item[(D1)] $\ph_1\Or\ph_2:= \Not(\Not \ph_1\And\Not\ph_2)$.
\eitem

For each observational proposition $\ph$, we assign its 
projection-valued truth value $\valo{\ph}\in\cQ(\cH)$ 
by the following rules \cite{BvN36}.
\bitem
\item[(T1)] $\valo{X\leo x}=E^{X}(x).$
\item[(T2)] $\valo{\Not \ph}=\valo{\ph}^{\perp}.$
\item[(T3)] $\valo{\ph_1\And \ph_2}=\valo{\ph_1}\And\valo{\ph_2}.$
\eitem
From (D1),  (T2) and (T3), we have
\bitem
\item[(D2)] $\valo{\ph_1\Or \ph_2}=\valo{\ph_1}\Or\valo{\ph_2}.$
\eitem

We define the {\em probability} $\Pr\{\ph\|\rh\}$ 
of an observational proposition $\ph$ in a state $\rh$ by
\bitem
\item[(P1)] $\Pr\{\ph\|\rh\}=\Tr[\valo{\ph}\rh]$.
\eitem
We say that {\em an observational proposition $\ph$ holds in a state $\rh$} iff
$\Pr\{\ph\|\rh\}=1$.

The standard interpretation of quantum theory restricts observational propositions
to be standard defined as follows.
\bitem 
\item[(W1)] An observational proposition including atomic formulas
$X_1\leo x_1, \ldots,X_n\leo x_n$ is called {\em standard} iff $X_1,\ldots,X_n$
are mutually commuting.   
\eitem 

All the standard observational propositions including only given mutually 
commuting observables 
$X_1, \ldots,X_n$ comprise a complete Boolean algebra under the logical order $\le$ defined by
$\ph\le \ph'$ iff $\valo{\ph}\le \valo{\ph'}$ and obey inference rules in classical logic.
Suppose that $X_1,\ldots,X_n\in\cO(\cM)$ are mutually commuting.
Let $x_1,\ldots,x_n\in\R$.
Then, 
$X_1\leo x_1\And\cdots\And X_n\leo x_n$ is a standard observational proposition.
We have
\beq
\valo{ X_1\leo x_1\And\cdots\And X_n\leo x_n}
=E^{X_1}(x_1)\And\cdots\And E^{X_n}(x_n)
=E^{X_1}(x_1)\cdots E^{X_n}(x_n).
\eeq
Hence, we reproduce the Born statistical formula as 
\beq
\Pr\{X_1\leo x_1\And\cdots\And X_n\leo x_n\|\rh\}
=\Tr[E^{X_1}(x_1)\cdots E^{X_n}(x_n)\rh].
\eeq
From the above, our definition of the truth values of observational propositions 
are consistent with the standard probabilistic interpretation of quantum theory. 

From Proposition \ref{th:QBorel} and (T1), we conclude 
\beq
\val{\tX\le\tilde{x}}=\valo{X\leo x}
\eeq
for all $X\in\cO(\cM)$ and $x\in\R$.
To every observational proposition $\ph$ the corresponding statement 
$\tilde{\ph}$ in $\cL(\in,\RQ)$ is given by the following rules for any 
$X\in\cO(\cM)$ and $x\in\R$, 
and observational propositions $\ph,\ph_1,\ph_2$.
\bitem
\item[(Q1)] $\displaystyle\widetilde{X\leo x}:=\tilde{X}\le\tilde{x}$.
\item[(Q2)] $\widetilde{\Not \ph}:=\Not\tilde{\ph}.$
\item[(Q3)] $\widetilde{\ph_1\And \ph_2}:={\tilde\ph_1}\And\tilde{\ph_2}.$
\eitem
Then, it is easy to see that the relation
\beq
\val{\tilde{\ph}}=\valo{\ph}
\eeq
holds for any observational proposition $\ph$.  
Thus, all the observational propositions are embedded in the set of 
statements in $\cL(\in,\RQ)$ with the same projection-valued truth value.

We denote by $\Sp(X)$ the spectrum of an observable $X\in\cO(\cM)$,
i.e., the set of all $\la\in\R$ such that $X-\la 1$ has a bounded inverse operator
on $\cH$.
An observable $X\in\cO(\cM)$ is called {\em finite} iff 
$\Sp(X)$ is a finite set, and {\em infinite} otherwise.
Denote by  $\cO_{\om}(\cM)$ is the set of finite observables in $\cO(\cM)$.

Let $X\in\cO_{\om}(\cM)$.  Then, $\Sp(X)$ coincides with the set of
eigenvalues of $X$.
Let 
\beq
\de(X)=\min_{x,y\in\Sp(X), x\not= y}\{|x-y|/2, 1\}.
\eeq
For any $x\in\R$,  we define the observational proposition $X=_{o}x$ by
\beq
X=_{o}x:=x-\de(X)<X\leo x+\de(X).
\eeq
Then, it is easy to see that we have  
\beq
\valo{X=_{o}x}=E^{X}(\{x\})
\eeq
for all $x\in\R$.

In Ref.~\cite{11QRM} we have introduced observational propositions 
for the case where $\dim(\cH)<\infty$ and $\cM=\cB(\cH)$ by 
rules (R'1), (R'2) of well-formed formulas and rules (T'1)--(T'3) for
projection-valued truth value assignment as follows.
\bitem
\item[(R'1)] For any $X\in\cO(\cB(\cH))$ and $x\in \R$, the expression
$X=_{o'} x$ is an observational proposition.
\item[(R'2)] If $\ph_1$ and $\ph_2$ are observational propositions,
$\Not \ph_1$ and $\ph_1\And \ph_2$ are also observational propositions.
\item[(T'1)] $\valoo{X=_{o'} x}=E^{X}(x).$
\item[(T'2)] $\valoo{\Not \ph}=\valoo{\ph}^{\perp}.$
\item[(T'3)] $\valoo{\ph_1\And \ph_2}
=\valoo{\ph_1}\And\valoo{\ph_2}.$
\eitem
Denote by $\Loo(\cB(\cH))$ the set of observational propositions
constructed by rules (R'1) and (R'2).
In this language, for any observables $X\in\cO(\cB(\cH))$ 
and any real number $x\in \R$, we can introduce the observational proposition 
$X\leoo x$ in $\Loo(\cB(\cH))$ by
\beq
X\leoo x:=\Sup_{x_j\in\Sp(X)\cap(-\infty, x]}X=_{o'}x_j,
\eeq
where the observational proposition $\Sup_{j}\ph_j$ is defined by
$\Sup_{j}\ph_j=\ph_1\Or\cdots\Or\ph_n$
for any finite sequence of observational propositions $\ph_1,\ldots,\ph_n$.
Then, we have
\beq	
\valoo{X\leoo x}=E^{X}(x).
\eeq
Now, we can conclude that if $\dim(\cH)<\infty$, the language $\Lo(\cB(\cH))$ and
$\Loo(\cB(\cH))$ are equivalent in the sense that there is a one-to-one 
correspondence $\Ph$ of $\Loo(\cB(\cH))$ onto $\Lo(\cB(\cH))$ 
such that $\valo{\Ph(\ph)}=\valoo{\ph}$, 
$\Ph(X=_{o'}x)=(X=_{o}x)$, and $\Ph(X\leoo x)=(X\leo x)$
for all $\ph\in\Loo(\cB(\cH))$, $X\in\cO(\cB(\cH))$, and $x\in\R$.
Thus, in what follows for the case where $\dim(\cH)<\infty$ we shall identify
the language $\Loo(\cB(\cH))$ introduce in Ref.~\cite{11QRM} 
with the language $\Lo(\cB(\cH))$; in this case we have $
\cO(\cB(\cH))=\cO_{\om}(\cB(\cH))$. 

\section{Simultaneous determinateness}
\label{se:3}

In this section, we shall examine basic properties of the commutator 
$\com(\tX_1,\ldots,\tX_n)$ for observables $X_1,\ldots,X_n\in\cO(\cM)$.
Let $X_1,\ldots,X_n\in\cO(\cM)$.
We denoted by
$\{X_1,\ldots,X_n\}''$  the von Neumann algebra generated by projections 
$E^{X_{j}}(\la)$ for all $j=1,\ldots,n$ and $\la\in\R$,
and denote by $\cZ(X_1,\ldots,X_n)$ the center of $\{X_1,\ldots,X_n\}''$,
i.e., $\cZ(X_1,\ldots,X_n)=\{X_1,\ldots,X_n\}''\cap \{X_1,\ldots,X_n\}'$.
The {\em cyclic subspace} $\cC(X_1,\ldots,X_n;\rh)$ of $\cH$ 
generated by $X_1,\ldots,X_n$, and $\rh$ is
defined by 
\[
\cC(X_1,\ldots,X_n;\rh)=\{X_1,\ldots,X_n\}''\overline{\ran}(\rh),
\]
where $\overline{\ran}$ stands for the closure of the range.
Then, $\cC(X_1,\ldots,X_n;\rh)$ is the least invariant subspace
under $\{X_1,\ldots,X_n\}''$ containing $\rh$.
Denote by  $C(X_1,\ldots,X_n;\rh)$ the projection of $\cH$ onto 
$\cC(X_1,\ldots,X_n;\rh)$.
Then, $C(X_1,\ldots,X_n;\rh)$ is the smallest projection $P$ in 
$\{X_1,\ldots,X_n\}'$ such that $P\rh=\rh$.

Under the Takeuti correspondence, the commutator of observables 
are characterized as follows.
\begin{theorem}\label{th:com_real}
For any $X_1,\ldots,X_n\in\cO(\cM)$, the following relations hold.
\bitem
\item $\com(\tX_1,\ldots,\tX_n)
=\cP\{\psi\in\cH\mid \mb{$[A,B]\psi=0$
for all $A,B\in\{X_1,\ldots,X_n\}''$}\}$.
\item $\com(\tX_1,\ldots,\tX_n)
=\cP\{\ps\in\cH\mid \mb{$[E^{X_j}(r_1),E^{X_k}(r_2)]E^{X_l}(r_3)\psi=0$}$\\
\hfill$\mb{ for all  $r_1,r_2,r_3\in\Q$ and $j,k,l=1,\ldots,n$}\}$.
\item $\com(\tX_1,\ldots,\tX_n)
=\max\{E\in\cP(\cZ(X_1,\ldots,X_n))\mid \mbox{$X_{j}E\commutes X_{k}E$}$\\
\hfill $\mb{for all $j,k=1,\ldots,n$}\}$.
\eitem
\end{theorem}
\begin{proof}
Let  $\cA=\bL(\tX_1,\ldots,\tX_n)$.
Then, $\com(\tX_1,\ldots,\tX_n)=\com(\cA)$.
We have
\[
\bL(\tX_1,\ldots,\tX_n)=\{E^{X_j}(r_j)\mid
r_j\in\Q \mb{ and }j=1,\ldots,n\}\cup\{0,1\},
\]
and hence $\bL(\tX_1,\ldots,\tX_n)''=\{X_1,\ldots,X_n\}''$.
Thus, relations (i) and (ii) follow from Theorem \ref{th:com} (i) and (ii),
respectively.
From Theorem \ref{th:equivalence_com} we have 
\beqas
\com(\tX_1,\ldots,\tX_n)
&=&\max\{E\in Z(\cA)\mid P_{1}\And E\commutes P_{2}\And E
\mb{ for all }P_{1},P_{2}\in\cA\}\\
&=&\max\{E\in\cP(\cZ(X_1,\ldots,X_n))\mid \mbox{$X_{j}E\commutes X_{k}E$
for all $j,k=1,\ldots,n$}\},
\eeqas
from which relation (iii) follows.
\end{proof}
\color{black}
We say that observables $X_1,\ldots,X_n\in\cO(\cM)$ are 
{\em simultaneously determinate} in a state $\rh$ iff $\Tr[\com(\tX_1,\ldots,\tX_n)\rh]=1$.

A probability measure $\mu$ on $\cB(\R^{n})$
is called a {\em joint probability distribution} of 
$X_1,\ldots,X_n\in\cO(\cM)$ in $\rh\in\cS(\cH)$ 
iff
for any polynomial $p(f_1(X_1),\ldots,f_n(X_n))$ of observables
$f_1(X_1),\ldots,f_n(X_n)$, where $f_1,\ldots,f_n\in B(\R)$,
we have
\beql{JPD-26}
\Tr[p(f_1(X_1),\ldots,f_n(X_n))\rh]
=
\idotsint_{\R^{n}}p(f_1(x_1),\ldots,f_n(x_n))\,
d\mu(x_1,\ldots,x_n).
\eeq

A joint probability distribution of $X_1,\ldots,X_n$ in $\rh$ is unique, if any.
Since simultaneous determinateness is considered to be a state-dependent 
notion of commutativity, it is expected that simultaneous determinateness
is equivalent to the state-dependent existence of the joint probability distribution.
This is indeed shown below together with other useful characterizations
of this notion.

\begin{theorem}\label{th:JPD}
For any observables $X_1,\ldots,X_n\in\cO(\cM)$ and a state $\rh\in\cS(\cH)$, 
the following conditions are all equivalent.
\bitem
\item $X_1,\ldots,X_n$ are simultaneously determinate in $\rh$, i.e., 
$\Tr[\com(\tX_1,\ldots,\tX_n)\rh]=1$
\item $\com(\tX_1,\ldots,\tX_n)\rh=\rh$.
\item $C(X_1,\ldots,X_n;\rh)\le\com(\tX_1,\ldots,\tX_n)$.
\item $[A,B]\rh=0$ for all $A,B\in\{X_1,\ldots,X_n\}''$.
\item There exists a joint probability distribution
of $X_1,\ldots,X_n$ in $\rh$. 
\item $X_j C(X_1,\ldots,X_n;\rh)\commutes X_k C(X_1,\ldots,X_n;\rh)$
for all $j,k=1,\ldots.n$.
\item There exists a spectral measure $E$ in $\cM$ on $\R^{n}$
satisfying
\beq
E(\De_1\times\cdots\times\De_n)\rh=E^{X_1}(\De_1)\And\cdots\And E^{X_n}(\De_n)\rh
\eeq
for all $\De_1,\ldots,\De_n\in\cB(\R)$.
\item There exists a probability measure $\mu$ on $\R^{n}$ satisfying
\beql{JPD-vi}
\mu(\De_1\times\cdots\times\De_n)=
\Tr[E^{X_1}(\De_1)\And \cdots\And E^{X_n}(\De_n)\rh]
\eeq
for any $\De_1,\ldots, \De_n\in\cB(\R)$.
\eitem
If one of the above conditions holds, the joint probability distribution $\mu$ of
of $X_1,\ldots,X_n$ in $\rh$ is uniquely determined by \Eq{JPD-vi}.
\end{theorem}

\begin{proof}Let $\cB=\{X_1,\ldots,X_n\}''$ and $C=C(X_1,\ldots,X_n;\rh)$.

(i)$\THEN$(ii): The assertion follows from the relation
$\|P\sqrt{\rh}-\sqrt{\rh}\|_{HS}^2=1-\Tr[P\rh]$ for any projection $P$, where
 $\|\cdots\|_{HS}$ is the Hilbert-Schmidt norm.

(ii)$\THEN$(iii): Since $\com(\tX_1,\ldots,\tX_n)\in\cB'$,
(iii) follows from (ii) by minimality of $C(X_1,\ldots,X_n;\rh)$.

(iii)$\THEN$(iv): It follows from (iii) that 
$\ran(\rh)\subseteq\ran(\com(\tX_1,\ldots,\tX_n))$ so that (iv) follows from 
Theorem \ref{th:com_real} (i).

(iv)$\THEN$(v): 
It follows from assumption (iv) and Proposition 2.2 in \Cite{HC99} that 
the GNS representation $(\cH,\pi,\Om)$ of $\cB$ induced by $\rh$ is abelian 
(i.e., $\pi(\cB)$ is abelian) and normal.  
Let $j=1, \ldots,n$.  Let $f_j$ be a bounded Borel function on $\R$.
By normality of $\pi$, there is a self-adjoint operator $\pi(X_j)$ affiliated with $\pi(\cB)$
such that $E^{\pi(X_j)}(\De)=\pi(E^{X_j}(\De))$ for all $\De\in\cB(\R)$,
and hence we have
\[
\pi(f_j(X_j))=f_j(\pi(X_j)). 
\]
Thus, the relation
\beqas
\mu(\De_1\times\cdots\times\De_n)
=
(\Om,E^{\pi(X_1)}(\De_1)\cdots E^{\pi(X_n)}(\De_n)\Om),
\eeqas
where $\De_1,\ldots,\De_n\in\cB(\R)$, defines a probability measure $\mu$
on $\cB(\R^n)$ satisfying
\[
\idotsint_{\R^{n}}p(f_1(x_1),\ldots,f_n(x_n))d\mu(x_1,\ldots,x_n)
=
(\Om,\pi\left(p\left(f_1(X_1),\ldots,f_n(X_n)\right)\right)\Om)
\]
for any polynomial $p\left(f_1(X_1),\ldots,f_n(X_n)\right)$ of 
$f_1(X_1),\ldots,f_n(X_n)$.
Thus, assertion (iv) follows from the relation
\beqas
\Tr[A\rh]
=
(\Om,\pi(A)\Om)
\eeqas
for any $A\in\cB$ satisfied by the GNS representation $(\cH,\pi,\Om)$.

(v)$\THEN$(i): Suppose that there exists a joint probability distribution 
$\mu$ of $X_1,\cdots X_n$ in $\rh$.
Then, for any $j,k,l=1,\ldots,n$ and $r_1,r_2,r_3\in\Q$,
we have
\[
\Tr[|[E^{X_j}(r_1),E^{X_2}(r_2)]E^{X_l}(r_3)|^2\rh]=0
\]
and we have $[E^{X_j}(r_1),E^{X_2}(r_2)]E^{X_l}(r_3)\rh=0$.
From Theorem \ref{th:com_real} (ii), it follows that 
$\com(\tX_1,\dots,\tX_n)\rh\psi=\rh\psi$ for all $\ps\in\cH$.  Thus, we have 
$\com(\tX_1,\dots,\tX_n)\rh=\rh$  and hence 
 $X_1,\ldots,X_n$ are simultaneously determinate in a state $\rh$.
 
(iii)$\IFF$(vi): From Theorem \ref{th:equivalence_com} (iv), we have 
\[
\com(\tX_1,\ldots,\tX_n)=\max\{E\in \cB'\mid X_j E\commutes X_k E ~~\mb{for all}~ j,k=1,\ldots,n\}.
\]
Thus, the equivalnce (iii)$\IFF$(vi) follows from the fact that $C(X_1,\ldots,X_n;\rh)\in\cB'$.

(ii)$\THEN$(vii): 
Let $G=\com(\tX_1,\ldots,\tX_n)$.
Then, we have $X_jG\commutes X_k G$ for all $j,k=1,\ldots,n$.
It follows that there exists a spectral measure $E$ in $\cM$ on $\R^{n}$
such that
\[
E(\De_1\times\cdots\times\De_n)
=E^{X_1G}(\De_1) \cdots E^{X_nG}(\De_n)
=E^{X_1G}(\De_1)\And \cdots \And E^{X_nG}(\De_n)
\]
for all $\De_1,\ldots,\De_n\in\cB(\R)$.
Let $\De_1,\ldots,\De_n\in\cB(\R)$.
Since 
\[
E^{X_jG}(\De)=E^{X_j}(\De)G+\de_{0}(\De)G\p
\]
for all $\De\in\cB(\R^n)$ and $j=1,\ldots,n$, 
where $\de_{0}$ is the point mass on $0\in\R$,
we have
\[
E(\De_1\times\cdots\times\De_n)G=
[E^{X_1}(\De_1)\And \cdots \And E^{X_n}(\De_n)]G
\]
If (ii) holds, i.e., $G\rh=\rh$, we have
\deqs{
E(\De_1\times\cdots\times\De_n)\rh
&=E(\De_1\times\cdots\times\De_n)G\rh\\
&=[E^{X_1}(\De_1)\And \cdots \And E^{X_n}(\De_n)]G\rh\\
&=E^{X_1}(\De_1)\And \cdots\And E^{X_n}(\De_n)\rh.}
Thus, (vii) follows.

(vii)$\THEN $(viii): Obvious.

(viii)$\THEN$(ii): Let $\mu$ be a probability measure on $\R^{n}$ 
satisfying \eq{JPD-vi}. Let $j,k,l\in\{1,\ldots,n\}$.  
By taking an appropriate marginal measure of $\mu$ there exists a probability
measure $\mu'$ on $\R^{3}$ such that  
\[
\mu'(\De_1\times\De_2\times\De_3)=\Tr[E^{X_j}(\De_1)\And E^{X_k}(\De_2)
\And E^{X_l}(\De_3)\rh]
\]
for all $\De_1,\De_2,\De_3\in\cB(\R)$.
Let $\De_1,\De_2,\De_3\in\cB(\R)$ and 
\beqas
P&=&E^{X_l}(\De_3)-E^{X_k}(\De_2^{c})\And E^{X_l}(\De_3)
-E^{X_j}(\De_1^{c})\And E^{X_k}(\De_2)\And E^{X_l}(\De_3)\\
& &
-E^{X_j}(\De_1)\And E^{X_k}(\De_2)\And E^{X_l}(\De_3),
\eeqas
where $\De^{c}$ stands for the complement of $\De\in\cB(\R)$.
Then, by the additivity of $\mu'$ we have 
\beqas
\Tr[P\rh]=\mu'(\R\times\R\times\De_3)-\mu'(\R\times\De_2^{c}\times\De_3)
-\mu(\De_1^{c}\times\De_2\times\De_3)
-\mu(\De_1\times\De_2\times\De_3)
=0.
\eeqas
Since $\Tr[(P\sqrt{\rh})^{\da}(P\sqrt{\rh})]=\Tr[P\rh]$, we have
$P\sqrt{\rh} =0$,  so that $E^{X_j}(\De_1)E^{X_k}(\De_2)P\rh =0$, and hence
$E^{X_j}(\De_1)E^{X_k}(\De_2)E^{X_l}(\De_3)\rh=
E^{X_j}(\De_1)\And E^{X_k}(\De_2)\And E^{X_l}(\De_3)\rh$.
By symmetry we also have 
$E^{X_k}(\De_2)E^{X_j}(\De_1)E^{X_l}(\De_3)\rh=
E^{X_j}(\De_1)\And E^{X_k}(\De_2)\And E^{X_l}(\De_3)\rh$.
Thus, we have $[E^{X_j}(\De_1),E^{X_k}(\De_2)]E^{X_l}(\De_3)\rh\ps=0$
for all $\ps\in\cH$.  Since $\De_1,\De_2,\De_3$ were arbitrary, 
it follows from Theorem \ref{th:com_real} that
$\ran(\rh)\subseteq\ran(\com(\tX_1,\ldots,\tX_n))$, and (ii) follows. 

Suppose that there exists a joint probability distribution $\mu$ of $X_1,\ldots,X_n$ in $\rh$. 
Let $f(x_1,\ldots,x_n)=\chi_{\De_{1}}(x_1)\cdots\chi_{\De_{n}}(x_n)$, where $\De_1,\ldots,
\De_n\in\cB(\R)$.
From \Eq{JPD-26}, we have
\deqs{
\Tr[E^{X_1}(\De_1)\cdots E^{X_n}(\De_n)\rh]
&=
\idotsint_{\R^{n}}\chi_{\De_{1}}(x_1)\cdots\chi_{\De_{n}}(x_n)d\mu(x_1,\ldots,x_n)\\
&=\mu(\De_1\times\cdots\times \De_n).
}
From the proof of (ii)$\THEN$(vii), we have
\deqs{
E^{X_1}(\De_1)\cdots E^{X_n}(\De_n)\rho
&=E^{X_1G}(\De_1) \cdots E^{X_nG}(\De_n)\rho\\\
&=E^{X_1G}(\De_1)\And \cdots \And E^{X_nG}(\De_n)\rho\\
&=E^{X_1}(\De_1)\And \cdots \And E^{X_n}(\De_n)\rho.
}
Thus, one of conditions (i)--(viii) holds, \Eq{JPD-vi} uniquely determines the 
joint probability distribution $\mu$ of $X_1,\ldots,X_n$ in $\rh$. 
\end{proof}

The equivalence between (i) and (v) in the above theorem was previously reported
in Theorem 2 of \Cite{11QRM} for the case where $\cM=\cB(\cH)$ with $\cH<\infty$.
The equivalence of (ii), (vii), and (viii) was given in Theorem 5.1 of \Cite{06QPC} for
the case $n=2$.

Note that for any $X_1,\ldots,X_n\in\cO(\cM)$
there exists a proposition $\ph$ in $\Lo(\cM)$ such that 
$\valo{\ph}=\com(\tX_1,\ldots,\tX_n)$,
since $\com(\tX_1,\ldots,\tX_n)\in\cO_{\om}(\cM)$ and 
$\valo{\com(\tX_1,\ldots,\tX_n)=_{o}1}=\com(\tX_1,\ldots,\tX_n)$.
However, it is not in general possible to construct such $\ph$ from atomic
propositions of the form $X_j\leo \la$ for $j=1,\ldots,n$ with $\la\in\R$.
In what follows, we shall show that this is possible for finite observables.

For any finite observables
$X_1,\ldots,X_n\in\cO_{\om}(\cM)$ we define the observational proposition
$\cm(X_1,\ldots,X_n)$ by
\beq
\cm(X_1,\ldots,X_n):=\Sup_{x_1\in\Sp(X_1),\ldots,x_n\in\Sp(X_n)} 
X_1=_{o}x_1\And\cdots\And X_n=_{o}x_n.
\eeq
Then, we have the following theorem.
\begin{theorem}\label{th:JD}
For any finite observables  $X_1,\ldots,X_n\in\cO_{\om}(\cM)$,
we have
\beql{JD}
\valo{\cm(X_1,\ldots,X_n)}=\com(\tX_1,\ldots,\tX_n).
\eeq
\end{theorem}
\begin{proof}
Let $X_1,\ldots,X_n\in\cO_{\om}(\cH)$.
Let $x^{(1)}_j<\cdots<x^{(n_j)}_j\in\R$ be the ascending 
sequence of eigenvalues of $X_j$.
Then, we have
\[
\bL(\tX_1,\ldots,\tX_n)=
\{E^{X_j}(x)\mid x=x^{(1)}_j,\ldots,x^{(n_j)}_j;\ j=1,\ldots,n\}\cup\{0\}.
\]
Since $\bL(X_1,\ldots,X_n)$ is a finite set, 
it is easy to see that the relations 
\beqas
\valo{\cm(X_1,\ldots,X_n)}
&=&\com(\{E^{X_j}(\{x\})\mid x=x^{(1)}_j,\ldots,x^{(n_j)}_j;\ j=1,\ldots,n\})\\
&=&\com(\bL(\tX_1,\ldots,\tX_n))\\
&=&\com(\tX_1,\ldots,\tX_n)
\eeqas
hold.
\end{proof}

The observational proposition $\cm(X_1,\ldots,X_n)$
was previously introduced in Ref.~\cite{11QRM} 
for the case where $\cM=\cB(\cH)$ and $\dim(\cH)<\infty$.
The following theorem is a straightforward generalization of 
Theorem 1 in Ref.~\cite{11QRM}. 
\begin{theorem}\label{th:SD}
Finite observables  $X_1,\ldots,X_n\in\cO_{\om}(\cM)$ are simultaneously determinate in
a vector state $\ps$  if and only if the state $\ps$ is a superposition
of common eigenvectors of $X_1,\ldots,X_n$.
\end{theorem}
\color{black}

\section{Quantum equality}
\label{se:4}
In this section, we shall examine basic properties of the $\cQ$-valued equality 
relation $\val{\tX=\tY}$ defined through $\VQ$ 
for any two observables $X,Y\in\cO(\cM)$, where $\cQ=\cP(\cM)$.
From Theorem 6.3 of Ref.~\cite{07TPQ}, we have the following characterizations.

\begin{theorem}\label{th:equality}
For any $X,Y\in\cO(\cM)$, the following relations hold.
\bitem
\item
$\val{\tX=\tY}=\cP\{\psi\in\cH\mid\mb{$E^{X}(r)\psi=E^{Y}(r)\psi$
for all $r\in\Q$}\}. $
\item
$\val{\tX=\tY}=\cP\{\psi\in\cH\mid\mb{$f(X)\psi=f(Y)\psi$ for all $f\in B(\R)$}\}.$
\item
$\val{\tX=\tY}=\cP\{\psi\in\cH\mid\mb{$(E^{X}(\De)\psi,E^{Y}(\Ga)\psi)=0$}$\\
\hfill$\mb{ for any $\De,\Ga\in\cB(\R)$ with $\De\cap\Ga=\emptyset$}\}.$
\eitem
\end{theorem}

\color{black}
We introduce a new atomic observational proposition $X=_{o}Y$ 
in $\Lo(\cM)$
for all $X,Y\in\cO(\cM)$ by the following additional rules for formation of observational
propositions and for projection-valued truth values:
\bitem
 \item[(R3)] For any $X,Y\in\cO(\cM)$, the expression
$X=_{o}Y$ is an observational proposition.
\item[(T4)]
$ \valo{X=_{o}Y}=\val{\tX=\tY}.
$
\eitem
We extend the correspondence between observational propositions 
and formulas in $\cL(\in,\VQ)$ by the following rule for any $X,Y\in\cO(\cM)$.
\bitem
\item[(Q4)] $\displaystyle\widetilde{X=_{o}Y}:=\tilde{X}=\tilde{Y}$.
\eitem
Then, from (T4) it is easy to see  that the relation
\beq\label{eq:OP-QST}
\val{\tilde{\ph}}=\valo{\ph}
\eeq
holds for any observational proposition $\ph$.  
We denote by $\Lo(\cM,=)$ the set of observational 
propositions constructed by rules (R1), (R2), (R3).
Then, the language $\Lo(\cM,=)$ is embedded in the set of 
statements in $\cL(\in,\RQ)$ by rules (Q1), (Q2), (Q3), (Q4)
with the same projection-valued truth value by rules (T1), (T2),
(T3), (T4).

In general, the equality relation in $\VQ$ is not an equivalence relation in $\VQ$
\cite{Ta81}.
From Theorem 6.3 of Ref.~\cite{07TPQ}, however, we conclude that
that $\cQ$-valued equality between two observables is indeed a
$\cQ$-valued equivalence relation as follows.

\begin{theorem}
For any observables $X,Y,Z\in\cO(\cM)$, the following relations hold.
\bitem
\item $\valo{X=_{o}X}=1$.
\item $\valo{X=_{o}Y}=\valo{Y=_{o}X}$.
\item $\valo{X=_{o}Y}\And \valo{Y=_{o}Z}\le \valo{X=_{o}Z}$.
\eitem
\end{theorem}
\color{black}

We say that observables $X$ and $Y$ are {\em equal in
a state $\rh$}, in symbols $X=_{\rh}Y$,  
iff $\Pr\{X=_{o}Y\|\rh\}=1$, or equivalently iff
$\valo{X=_{o}Y}\rh=\rh$.  
In general, we say that 
observables $X$ and $Y$ are {\em equal in a state $\rh$
with probability $\Pr\{X=_{o}Y\|\rh\}$}.
On the other hand, we have explored another relation called quantum 
perfect correlation in Ref.~\cite{06QPC} as follows.
Two observables $X$ and $Y$ are called {\em perfectly correlated} 
in a state $\rh$ iff $\Tr[E^{X}(\De)E^{Y}(\Ga)\rh]=0$ 
for any disjoint Borel sets $\De,\Ga\in\cB(\R)$.
It is noted that the quantity 
$\Tr[E^{X}(\De)E^{Y}(\Ga)\rh]=0$ for $\De,\Ga\in\cB(\R)$
is called the {\em weak joint distribution} of $X$ and $Y$ in $\rh$,
and known to be experimentally accessible by weak measurement and
post-selection \cite{11UUP}.
We shall show that the above two relations are equivalent 
together with other equivalent conditions
to conclude that the relation $X=_{\rh}Y$ and the 
probability $\Pr\{X=_{o}Y\|\rh\}$
are experimentally accessible.

\begin{theorem}\label{th:ID}
For any observables $X,Y\in\cO(\cM)$ and $\rh\in\cS(\cH)$, 
the following conditions are all equivalent.
\bitem
\item $X=_{\rh}Y$, i.e., $\valo{X=_{o}Y}\rh=\rh$.
\item $X$ and $Y$ are perfectly correlated in $\rh$, i.e., 
$\Tr[E^{X}(\De)E^{Y}(\Ga)\rh]=0$ 
for all $\De,\Ga\in\cB(\R)$ with $\De\cap\Ga=\emptyset$.
\item $X\ps=Y\ps$ for all $\ps\in\cC(X,\rh)$.
\item $\bracket{\ps,E^{X}(\De)\ps}=\bracket{\ps,E^{Y}(\De)\ps}$
for all $\ps\in\cC(X,\rh)$.
\item $E^{X}(\De)\rh=E^{Y}(\De)\rh$ for all $\De\in\cB(\R)$.
\item $f(X)C(X;\rh)=f(Y)C(X;\rh)$ for all $f\in B(\R)$.
\item $C(X;\rh)=C(Y;\rh)$ and $XC(X;\rh)=YC(X;\rh)$.
\item There exists a joint probability distribution  
$\mu^{X,Y}_{\rho}(x,y)$ of $X,Y$ in $\rho$ that satisfies
\beq
\mu^{X,Y}_{\rho}(\{(x,y)\in\R^{2}\mid x=y\})=1.
\eeq
\eitem
\end{theorem}
\begin{proof} The assertions follow from Theorem \ref{th:equality} above
and Theorems 3.2, Theorem 3.4, Theorem 4.3, and Theorem 5.3 in \Cite{06QPC}.
\end{proof} 
\color{black}

The equivalence between (i) and (viii) was previously reported in Theorem 4 in \Cite{11QRM}
for the case where $\cH<\infty$ and $\cM=\cB(\cH)$.

Let $\ph(X_1,\ldots,X_n)$ be an observational proposition 
that is constructed by rules (R1), (R2), (R3) and 
includes symbols for observables only from the list $X_1,\ldots,X_n$, {\em i.e., }
$\ph(X_1,\ldots,X_n)$ includes only atomic observational propositions 
of the form $X_j\leo x_j$ or $X_j=X_k$, where $j, k=1,\ldots,n$ and $x_j$ is
the symbol for an arbitrary real number.
In this case, $\ph(X_1,\ldots,X_n)$ is said to be an observational proposition
in $\Lo(X_1,\ldots,X_n)$.
Then, $\ph(X_1,\ldots,X_n)$ is said to be {\em contextually well-formed} 
in a state $\rh$ iff $X_1,\ldots,X_n$ are simultaneously determinate in $\rh$.
The following theorem answers the question as to in what state $\rh$  
the probability assignment satisfies rules for calculus of classical probability,  and shows 
that for well-formed observational propositions $\ph(X_1,\ldots,X_n)$ 
the projection-valued truth value assignment satisfies Boolean inference rules.
\begin{theorem}
Let  $\ph(X_1,\ldots,X_n)$ be an observational proposition in $\Lo(X_1,\ldots,X_n)$.
If $\ph(X_1,\ldots,X_n)$ is a tautology in classical logic,
then we have 
\[
\com(\tX_1,\ldots,\tX_n)\le\valo{\ph(X_1,\ldots,X_n)}.
\]
Moreover, if $\ph(X_1,\ldots,X_n)$ is contextually well-formed
in a state $\rh$, then $\ph(X_1,\ldots,X_n)$ holds in $\rh$.
\end{theorem}
\begin{proof}
Suppose that an observational proposition 
$\ph=\ph(X_1,\ldots,X_n)$ is a tautology in classical logic.
Let $\tilde{\ph}$ be the corresponding formula in $\cL(\in,\VQ)$.
Then, it is easy to see that there is a formula $\ph_0(u_1,\ldots,u_n, v_1,\ldots,v_m)$
in $\cL(\VQ)$ provable in ZFC satisfying
$\ph_0(\tX_1,\ldots,\tX_n,\tilde{r}_1,\ldots,\tilde{r}_m)=\tilde{\ph}$
with some real numbers $r_1,\ldots,r_m$.
Then, by the ZFC Transfer Principle (Theorem \ref{th:TP}), we have  
$\com(\tX_1,\ldots,\tX_n)\le\val{\tilde{\ph}}$.
Thus, the assertion follows from relation (\ref{eq:OP-QST}).
\end{proof}

The above theorem was previously announced as Theorem 3 in \Cite{11QRM}
for the case where $\cH<\infty$ and $\cM=\cB(\cH)$, the proof of which
needs quantum set theory and the embedding of the language of observational
propositions into the language of quantum set theory developed in this paper.

Note that for any $X,Y\in\cO(\cM)$
there exists a proposition $\ph$ in $\Lo(\cM)$ such that 
$\valo{\ph}=\val{\tX=\tY}$.
In fact,  we have $\val{\tX=\tY}\in\cO_{\om}(\cM)$ and 
$\val{\val{\tX=\tY}=_{o}1}=\val{\tX=\tY}$.
However, it is not in general possible to construct such $\ph$ from atomic
propositions of the form $X_j\leo \la$ for $j=1,\ldots,n$ with $\la\in\R$.
In what follows, we shall show that this is possible for finite observables.

For any finite observables $X,Y$, we define 
the observational proposition $X=Y$ by 
\beq
X=_{o}Y:=\Sup_{x\in \Sp(X)}X=_{o}x\And Y=_{o}x.
\eeq
Then, we have the following.
\begin{theorem}\label{th:EQ}
For any finite observables $X,Y\in\cO_{\om}(\cH)$,
we have
\beql{QE}
\valo{X=_{o}Y}=\val{\tX=\tY}.
\eeq
\end{theorem}
\begin{proof}
Let $\psi\in\cR(\valo{X=Y})$.
Then, for any $x\in\Sp(X)$, we have 
\[
E^{X}(\{x\})\psi=E^{X}(\{x\})\cap E^{Y}(\{x\})\psi=E^{Y}(\{x\})\psi,
\]
and for any $x\not\in \Sp(X)$, we have 
$E^{X}(\{x\})\psi=0=E^{Y}(\{x\})\psi.
$
Thus, $\psi\in\cR(\val{\tX=\tY})$ follows from Theorem \ref{th:equality} (i).
Conversely, suppose $\psi\in\cR(\val{\tX=\tY})$.
Then, for all $x\in\R$, we have
$E^{X}(\{x\})\psi=E^{Y}(\{x\})\psi$
so that we have  
$E^{X}(\{x\})\And E^{Y}(\{x\})\psi=E^{X}(\{x\})\psi$.
Thus, we have $\valo{X=Y}\psi=\psi$.  Therefore, the assertion follows.
\end{proof}

As shown in \Cite{11QRM} for the finite dimensional case, state-dependent
equality between finite observables are generally characterized in terms 
of eigenvectors as follows.
\begin{theorem}\label{th:EQ}
Finite observables  $X$ and $Y$ are equal in 
a vector state $\ps$  if and only if the state $\ps$ is a superposition
of common eigenvectors of $X$ and $Y$ with common eigenvalues.
\end{theorem}
\color{black}

\section{Measurements of observables}
\label{se:mob}

In this and next sections, we shall discuss measurements for a quantum system  
described by a von Neumann algebra $\cM$ on a Hilbert space $\cH$.

A {\em probability operator-valued measure (POVM)} for a 
von Neumann algebra $\cM$
on $\R^n$ is a mapping $\Pi:\cB(\R^{n})\to\cM$ satisfying the following conditions.
\bitem
\item[(M1)]  
$\Pi(\De)\ge0$ for all $\De\in\cB(\R^n)$.
\item[(M2)] $\sum_j \Pi(\De_j)=1$ for any disjoint sequence of Borel sets 
$\De_1,\De_2, \ldots\in\cB(\R^n)$ such that $\R^n=\bigcup_j \De_j$.
\eitem
 A {\em measuring process} for $\cM$ is defined to be a quadruple
$(\cK,\si,U,M)$ consisting of a Hilbert space $\cK$, a state (density operator)
$\si$ on $\cK$,  a unitary operator $U$ on $\cH\otimes\cK$, and an observable 
$M$ on $\cK$ satisfying
\beql{MP}
\Tr_{\cK}[U^{\da}(X\otimes E^{M}(\De))U(1\otimes \si)]\in\cM
\eeq
for every $X\in\cM$ and $\De\in\cB(\R)$, where  $\Tr_{\cK}$ stands for the 
partial trace on $\cK$ \cite{84QC,16A1}.

A measuring process $\M(\bx)=(\cK,\si,U,M)$ with {\em output variable} $\bx$
describes a
measurement carried out by an interaction, called the {\em measuring interaction}, 
from time 0 to time $\De t$ between the measured system
$\bS$ described by $\cM$ and the {\em probe} system $\bP$ described by $\cB(\cK)$ 
that is prepared in the state $\si$ at time 0.
The outcome of this measurement is obtained by
measuring the observable $M$, called the {\em meter observable},  
in the probe at time $\De t$.
The unitary operator $U$ describes the time evolution of $\bS+\bP$ from time 0 to
$\De t$.  We shall write $M(0)=1\otimes M$, $M(\De t)=U^{\da}M(0)U$, 
$X(0)=X\otimes 1$,  and $X(\De t)=U^{\da}X(0)U$
for any observable $X\in\cO(\cM)$.
We can use the probabilistic interpretation for the system $\bS+\bP$.
The {\em output distribution}
$\Pr\{\bx\in \De\|\rho\}$,
the probability distribution 
of the output variable $\bx$ of this measurement on input state $\rh\in\cS(\cH)$, 
is naturally defined as
\[
\Pr\{\bx\in \De\|\rho\}=\Pr\{M(\De t)\in \De\|\rho\otimes\si\}
=\Tr[E^{M(\De t)}(\De)\rho\otimes\si].
\]
The {\em POVM of the measuring process $\M(\bx)$} is defined by
\[
\Pi(\De)
=\Tr_{\cK}[E^{M(\De t)}(\De)(I\otimes\si)].
\]
Then,  $\Pi(\De)\in\cM$ for all $\De\in\cB(\R)$ by \Eq{MP} and
$\Pi:\cB(\R)\to\cM$ is a POVM  for $\cM$ on $\R$ satisfying
\bitem
\item[(M3)] $\Pr\{\bx\in \De\|\rh\}=\Tr[\Pi(\De)\rho]$.
\eitem
Conversely, from a general result  in \Cite{84QC} it can be easily seen
that for every POVM $\Pi$ for $\cM$ on $\R$,
there is a measuring process $\M(\bx)=(\cK,\si,U,M)$ for $\cM$ satisfying (P3).
In fact, for any fixed $\rh_0\in\cS(\cH)$ the relation 
$\cI(\De)^{*}X=\Tr[X\rh_{0}]\Pi(\De)$ for all $X\in\cM$ and $\De\in\cB(\R)$
defines a completely positive instrument for $\cB(\cH)$ on $\R$, and by Theorem 5.1
in \Cite{84QC} there exists a measuring process $\bM(\bx)=(\cK,\si,U,M)$ for $\cB(\cH)$
such that $\Tr[X\rh_{0}]\Pi(\De)=\Tr_{\cK}[U^{\da}(X\otimes E^{M}(\De))U(1\otimes\si)]$
for all $X\in\cM$ and $\De\in\cB(\R)$.
Then, it is easy to see that $\M(\bx)$ is a measuring process for $\cM$ and
satisfies (P3).
For further accounts of the universality of the class of measurement models 
described by measuring processes we refer the reader to \Cite{84QC} for
quantum systems with finite degrees of freedom and to \Cite{16A1} for those
with infinite degrees of freedom.

Let $A\in\cO(\cM)$ and $\rh\in\cS(\cH)$.
A measuring process $\M(\bx)=(\cK,\si,U,M)$ for $\cM$ with the POVM $\Pi$
is said to {\em measure $A$} in $\rho$
if $A(0)=_{\rho\otimes\si}M(\De t)$,
and {\em weakly measure} $A$ in $\rho$ iff
$
\Tr[\Pi(\De)E^{A}(\Ga)\rho]=\Tr[E^{A}(\De\cap\Ga)\rho]
$
for any $\De,\Ga\in\cB(\R)$.
A measuring process $\M(\bx)$ is said to 
{\em satisfy the Born statistical formula} (BSF)  for $A$
in $\rho$ iff it satisfies
$
\Pr\{\bx\in \De\|\rho\}=\Tr[E^{A}(\De)\rho]
$
for all $x\in\R$.
The following theorem characterizes measurements 
of an observable in a given state \cite{06QPC}.

\begin{theorem}\label{th:MOB}
Let $\M(\bx)=(\cK,\si,U,M)$ be a measuring process for $\cM$ with the POVM $\Pi(\De)$.
For any observable $A\in\cO(\cM)$ and any state $\rh\in\cS(\cH)$,
the following conditions are all equivalent.
\bitem
\item $\M(\bx)$ measures $A$ in $\rho$.
\item $\M(\bx)$ weakly measures $A$ in $\rho$.
\item $\M(\bx)$ satisfies the BSF for $A$ in any vector state
 $\psi\in\cC(A,\rho)$.
 \eitem
\end{theorem}

In the conventional approach, a measuring process $\M(\bx)=(\cK,\si,U,M)$ 
with the POVM $\Pi$ is considered to be a measurement of an observable $A$ iff $\Pi=E^{A}$ \cite{84QC}, 
since in this case the probability distribution of $A$ predicted 
by the Born formula is reproduced by the probability distribution of $\Pi$ in any state.  
However, in this approach it is not clear whether a measurement of an observable
$A$ actually reproduces the value of the observable $A$ just before the measurement. 
The following theorem, which is an immediate consequence of Theorem \ref{th:MOB}, 
ensures that this is indeed the case (cf.~the remark after Theorem 8.2 in \Cite{06QPC}).

\begin{theorem}
Let $\M(\bx)=(\cK,\si,U,M)$ be a measuring process for $\cM$ with the POVM $\Pi$.
Then, $\M(\bx)$ measures $A\in\cO(\cM)$ in any $\rh\in\cS(\cH)$
if and only if\/ $\Pi=E^{A}$.
\end{theorem}

\section{Simultaneous measurability}
\label{se:8}

For any measuring process $\M(\bx)=(\cK,\si,U,M)$ for $\cM$
and a real-valued Borel function $f$, 
the measuring process $\M(f(\bx))$ with output variable 
$f(\bx)$ is defined by $\M(f(\bx))=(\cK,\si,U,f(M))$.  Observables $A,B$ are 
said to be {\em simultaneously measurable} in a state $\rh\in\cS(\cH)$
by $\M(\bx)$ iff there are Borel functions $f,g$ such that
$\M(f(\bx))$ and $\M(g(\bx))$ measure $A$ and $B$ in $\rh$,
respectively.
Observables $A,B$ are said to be {\em simultaneously measurable}
 in $\rh$ iff there is a measuring process $\M(\bx)$ such that 
$A$ and $B$ are simultaneously measurable in $\rh$ by $\M(\bx)$.

Simultaneous measurability and simultaneous determinateness are not
equivalent notions under the state-dependent formulation, as
the following theorem clarifies; the case where $\dim(\cH)<\infty$
was previously reported in Ref.~\cite{11QRM}, Theorem10.

\begin{theorem}\label{th:main}
\bitem
\item Two observables $A,B\in\cO(\cM)$ are 
simultaneously determinate in a state $\rh\in\cS(\cH)$ 
if and only if there exists a POVM $\Pi$ for $\cM$ on $\R^{2}$ satisfying
\beqa
\Pi(\De\times\R)&=&E^{A}(\De)\quad\mbox{on $\cC(A,B,\rho)$ for all $\De\in\R$},
\label{eq:i-1}\\
\Pi(\R\times\Ga)&=&E^{B}(\Ga)\quad\mbox{on $\cC(A,B,\rho)$ for all $\Ga\in\R$}.
\label{eq:i-2}
\eeqa
\item Two observables $A,B\in\cO(\cM)$ are
simultaneously measurable in a state $\rh\in\cS(\cH)$
if and only if there exists a POVM $\Pi$ for $\cM$ on $\R^{2}$ satisfying
\beqa
\Pi(\De\times\R)&=&E^{A}(\De)\quad\mbox{on $\cC(A,\rho)$ for all $\De\in\R$},
\label{eq:ii-1}\\
\Pi(\R\times\Ga)&=&E^{B}(\Ga)\quad\mbox{on $\cC(B,\rho)$ for all $\Ga\in\R$}.
\label{eq:ii-2}
\eeqa
\item Two observables $A,B\in\cO(\cM)$ are
simultaneously measurable in a state $\rh\in\cS(\cH)$
if they are simultaneously determinate in $\rh$. 
\eitem
\end{theorem}
\begin{proof}
Let $\cC=\cC(A,B,\rh)$ and $C=C(A,B;\rh)$.

(i) (only if part):
Let $G=\com(\tA,\tB)$. Then, $G\in\cM$ and 
$AG\commutes BG$.
Let $\Pi$ be the joint spectral measure of $AG$ and $BG$, i.e.,
$\Pi(\De\times\Ga)=E^{AG}(\De)E^{BG}(\Ga)$ for all $\De,\Ga\in\cB(\R)$.
Then, $\Pi$ is a POVM for $\cM$ on $\R^{2}$.
Suppose that $A$ and $B$ are simultaneously determinate in a state $\rh$.
Then, $\ran(\rh)\subseteq\ran(\com(\tA,\tB))$.
By the minimality of $C(A,B,\rho)$ among $(A,B)$-invariant subspaces,
we have $C\le G$ and $AG,BG\commutes C$.  Thus, we have 
$\Pi(\De\times\R)C=E^{AG}(\De)C=E^{AC}(\De)C=E^{A}(\De)C$
and similarly $\Pi(\R\times\Ga)C=E^{B}(\Ga)C$ for all $\De,\Ga\in\cB(\R)$.  
Thus, $\Pi$ satisfies Eqs.~(\ref{eq:i-1}) and (\ref{eq:i-2}).

(i) (if part): 
Let  $\Pi$ be a POVM for $\cM$ on $\R^2$ satisfying (\ref{eq:i-1}) and (\ref{eq:i-2}).
Let $\Pi'$ be a positive operator valued measure for $\cB(\cH)$ on $\R^{2}$ defined by
$\Pi'(\De\times\Ga)=C\Pi(\De\times\Ga)C$ for all $\De,\Ga\in\cB(\R)$. 
Let $\Pi''$ be a POVM for $\cB(\cC)$ on $\R^{2}$ obtained by restricting 
$\Pi'$ to $\cC$.
Let $\De,\Ga\in\cB(\R)$.
By the definition of $C$, we have 
$E^{A}(\De)C=CE^{A}(\De)=CE^{A}(\De)C$ 
and $E^{A}(\De)C$ is a projection.
Similarly,  $E^{B}(\Ga)C=CE^{B}(\Ga)=CE^{B}(\Ga)C$ 
and $E^{B}(\Ga)C$ is a projection.
Thus, we have $\Pi''(\De\times\R)=C\Pi(\De\times\R)C=E^{A}(\De)C$,
and similarly  $\Pi''(\R\times\Ga)=C\Pi(\R\times\Ga)C=E^{B}(\Ga)C$.
Since $\De$ and $\Ga$ were arbitrary,
the marginals of $\Pi''$ are projection-valued.
By a well-know theorem (e.g.,  \Cite{Dav76}, Theorem 3.2.1),
the marginals commute and $\Pi''$ is the product of their marginals.
Thus, we have $AC\commutes BC$, and hence 
by Theorem \ref{th:JPD}, $A$ and $B$ are simultaneously determinate.

(ii) (only if part):
Suppose that $A,B\in\cO(\cM)$ are
simultaneously measurable in $\rh\in\cS(\cH)$.
Then, we have a measuring process $\M(\bx)=(\cK,\si,U,M)$ for $\cM$
and real-valued Borel functions
$f,g$ such that $\M(f(\bx))$ measures $A$ in $\rh$ and 
$\M(g(\bx))$ measures $B$ in $\rh$.
Let $\Pi_0$ be the POVM of  $\M(\bx)$.
Let $\Pi$ be a POVM on $\R^2$ such that $\Pi(\De\times\Ga)=\Pi_0(f^{-1}(\De)\cap g^{-1}(\Ga))$.
Then, it is easy to see that $\Pi$ satisfies Eqs.~(\ref{eq:ii-1}) and (\ref{eq:ii-2}).

(ii) (if part)
Let $\Pi$ be a  POVM  for $\cH$ on $\R^{2}$ satisfying Eqs.~(\ref{eq:ii-1}) and (\ref{eq:ii-2}).
Then, by the remark after condition (M3) in Section \ref{se:mob} 
there exists a measuring process 
$\M(\bx)=(\cK,\si,U,M)$ for $\cM$ and real-valued Borel functions
$f,g$ such that 
\beqa
\Pi(\De\times\Ga)=\Tr_{\cK}
[U^{\da}(I\otimes E^{f(M)}(\De)E^{g(M)}(\Ga))U(I\otimes \si)].
\eeqa
Then, we have 
\beqa
\Pi(\De\times\R)
=\Tr_{\cK}
[U^{\da}(I\otimes E^{f(M)}(\De))U(I\otimes \si)],
\eeqa
so that from \Eq{ii-1} we have $\M(f(\bx))$ measures $A$ in $\rh$.
Similarly, we can show that $\M(g(\bx))$ measures $B$ in $\rh$.

Assertion (iii) follows from (i) and (ii).
\end{proof}

Discussions on physical significance of the state-dependent formulation 
of simultaneous measurability have been given in \Cite{11QRM} for the finite dimensional case.
Further discussions on the state-dependent formulation of quantum measurement theory
will appear elsewhere.

\section{Conclusion}
Quantum set theory originated from the method of forcing introduced by Cohen \cite{Coh63,Coh66} 
for the independence proof of the continuum hypothesis and from quantum logic introduced 
by Birkhoff and von Neumann \cite{BvN36} for logical 
axiomatization of quantum mechanics.
After Cohen's work, Scott and Solovay \cite{SS67} reformulated the forcing 
method by Boolean-valued models of set theory \cite{Bel85}, which have become
a central method in the field of axiomatic set theory.
In 1978 Takeuti \cite{Ta78} started Boolean-valued analysis, which provides
systematic applications of logical meta-theorems for Boolean-valued models
to {\em not meta-}mathematical problems mainly in analysis.  
Among others,  Boolean-valued analysis made a
great successes in operator algebras \cite{Ta83a,Ta83b,83BH} and especially 
in solving a long-standing open problem in the structure theory of type I algebras 
applying the forcing method for cardinal collapsing  \cite{83BT,84CT,85NC}.

As a successor of those attempts, quantum set theory, 
a set theory based on the Birkhoff-von Neumann quantum logic,
was introduced by Takeuti \cite{Ta81}, who established the one-to-one correspondence 
between reals in the model (quantum reals)  and quantum observables.
Quantum set theory was recently developed by the present author \cite{07TPQ,14QST}
to obtain the transfer principle to determine quantum truth values of theorems 
of the ZFC set theory,  and to clarify the operational meaning 
of equality between quantum reals, which extends the probabilistic 
interpretation of quantum theory,

\color{black}
To formulate the standard probabilistic interpretation of quantum theory,
we have introduced the language of observational propositions with rules (R1) and (R2) 
for well-formed formulas constructed from atomic formulas of the form $X\leo x$, 
rules  (T1), (T2), and (T3) for projection-valued truth 
value assignment, and rule (P1) for probability assignment.
Then, the standard probabilistic interpretation gives the statistical predictions
for standard observational propositions formulated by (W1), 
which concern only a commuting family of observables.
The Born statistical formula is naturally derived in this way.
We have extended the standard interpretation by introducing the notion of 
simultaneous determinateness and atomic formulas of the form $X=Y$ for equality.  
To extended observational propositions 
formed through rules (R1), $\ldots$, (R4),  the projection-valued 
truth values are assigned by rule (T1), $\ldots$, (T4), and the probabilities
are assigned by rule (P1).  
Then, we can naturally extend the standard interpretation to a general
and state-dependent interpretation for observational propositions including
the relations of simultaneous determinateness and equality.
Quantum set theory ensures that any contextually well-formed formula
provable in ZFC has the probability assigned to be 1.
This extends the classical inference for quantum theoretical predictions 
from commuting observables to simultaneously determinate observables.
We apply this new interpretation to construct a theory of measurement of 
observables in the state-dependent approach, 
to which the standard interpretation cannot apply.
We have reported only basic formulations here, but further development 
in this approach will be reported elsewhere.

\section*{Acknowledgments}
The author thanks Minoru Koga for useful comments.
The author thanks the anonymous referee for their comments to improve the earlier version.
This work was supported by JSPS KAKENHI, No.~26247016 and No.~15K13456, 
and the John Templeton Foundation, ID \#35771.

\providecommand{\bysame}{\leavevmode\hbox to3em{\hrulefill}\thinspace}

\end{document}